%% file: main.tex
\documentclass[sigconf,nonacm]{acmart}
 
\usepackage{amssymb,amsthm,amsmath}
\usepackage{mathtools}
\usepackage{paralist}
\usepackage{xspace}

\usepackage{booktabs} 
\usepackage{stmaryrd}

\setcopyright{rightsretained}

\acmDOI{}

\acmISBN{}

\acmConference[PODS 2019]{the 38th ACM Symposium on Principles of Database Systems}{July
  2019}{Amsterdam, The Netherlands}
\acmYear{2019}
\copyrightyear{2018}

\acmArticle{}
\acmPrice{}

\editor{}
\editor{}
\editor{}

\input{macros}

\begin{document}
\title{Ontology Focusing: Know\-ledge-enriched Databases on Demand}



\author{Tomasz Gogacz}
\affiliation{%
  \institution{University of Warsaw}
}
\email{t.gogacz@mimuw.edu.pl}

\author{V\'ictor Guti\'errez-Basulto}
\affiliation{%
  \institution{Cardiff University}
}
\email{gutierrezbasultov@cardiff.ac.uk}

\author{Yazm\'in
Ib\'a\~nez-Garc\'ia}
\affiliation{%
  \institution{TU Wien}
}
\email{yazmin.garcia@tuwien.ac.at }

\author{Filip Mur\-lak}
\affiliation{%
  \institution{University of Warsaw}
}
\email{fmurlak@mimuw.edu.pl}

\author{Magdalena Ortiz}
\affiliation{%
 \institution{TU Wien}
}
\email{ortiz@kr.tuwien.ac.at}

\author{Mantas \v{S}imkus}
\affiliation{%
  \institution{TU Wien}
}
\email{simkus@dbai.tuwien.ac.at}

\renewcommand{\shortauthors}{T. Gogacz et al.}

\begin{abstract}
We propose a novel framework to facilitate the on-demand design of data-centric systems by exploiting domain knowledge from an existing ontology. Its key ingredient is a process that we call focusing, which allows to obtain a schema for a (possibly knowledge-enriched) database semi-automatically, given an ontology and a specification of the scope of the desired system. We formalize the inputs and outputs of focusing, and identify relevant computational problems: finding a schema via focusing, testing its consistency, and answering queries in the knowledge-enriched databases it produces. These definitions are fully independent from the ontology language. We then instantiate the framework using selected description logics as ontology languages, and popular classes of queries for specifying the scope of the system. For several representative combinations, we study the decidability and complexity of the identified computational problems. As a by-product, we isolate (and solve) variants of classical decision problems in description logics, that are  interesting in their own right.
\end{abstract}

%
%



\maketitle

 \input{1-introduction}


 \input{3-focusing}

\input{5-concrete-problems}

\input{6-discussion}

\bibliographystyle{ACM-Reference-Format}


\appendix
\onecolumn

\input{appendix-emptiness}

\end{document}

%% file: macros.tex
\newcommand{\mn}[1]{\ensuremath{\mathsf{#1}}}

\newcommand{\saveproblem}{\vspace{-3pt}}

\newcommand{\nexptime}{\textsc{NExpTime}\xspace}
\newcommand{\conexptime}{\textsc{coNExpTime}\xspace}
\newcommand{\exptime}{\textsc{ExpTime}\xspace}
\newcommand{\ptime}{\textsc{PTime}\xspace}
\newcommand{\twoexptime}{\textsc{2ExpTime}\xspace}

\newcommand{\np}{\textsc{NP}\xspace}
\newcommand{\conp}{\textsc{coNP}\xspace}

\newcommand{\nlogspace}{\textsc{NL}\xspace}

\newcommand{\Ff}{{\mathcal{F}}}
\newcommand{\Ii}{{\mathcal{I}}}
\newcommand{\Jj}{{\mathcal{J}}}

\newcommand{\Mm}{{\mathcal{M}}}

\newcommand{\Oo}{{\mathcal{O}}}

\newcommand{\Tt}{{\mathcal{T}}}

\newcommand{\dlname}[1]{\ensuremath{\mathcal{#1}}\xspace}

\newcommand{\ALCI}{\dlname{ALCI}}
\newcommand{\ALCO}{\dlname{ALCO}}

\newcommand{\ELIbot}{\dlname{ELI}_{\!\bot}}

\newcommand{\ALCOI}{\dlname{ALCOI}}
\newcommand{\ALCIF}{\dlname{ALCIF}}

\newcommand{\ALCHIF}{\dlname{ALCHIF}}
\newcommand{\ALCHOI}{\dlname{ALCHOI}}
\newcommand{\ALCHOIF}{\dlname{ALCHOIF}}
\newcommand{\DLLITEBOOL}{\mathit{DL\mbox{-}Lite}_{\mathrm{Bool}}^{\mathcal{H}\mathcal{O}\mathcal{F}}}

\newcommand{\DLLITEOBOOL}{\mathit{DL\mbox{-}Lite}_{\mathrm{Bool}}^{\mathcal{O}}}
\newcommand{\DLLITEFUNC}{\mathit{DL\mbox{-}Lite}^{\mathcal{H}\mathcal{F}}}
\newcommand{\DLLITE}{\mathit{DL\mbox{-}Lite}}

\newcommand{\ISA}{\ensuremath{\sqsubseteq}}

\newcommand{\tp}{\mn{tp}}

\newcommand{\adom}[1]{\mn{adom}(#1)}

\newtheorem{definition}{Definition}
\newtheorem{theorem}{Theorem}

\newtheorem{corollary}{Corollary}
\newtheorem{lemma}{Lemma}
\newtheorem{example}{Example}

\newcommand{\fix}{\mathsf{FIX}}
\newcommand{\cwa}{\mathsf{CL}}
\renewcommand{\det}{\mathsf{DET}}
\newcommand{\rest}{\cwa}

\newcommand{\Qcwa}{Q_{\cwa}}
\newcommand{\Qdet}{Q_{\det}}
\newcommand{\Qfix}{Q_{\fix}}

\newcommand{\Lth}{\mathcal{L}_{\mathsf{TH}}}
\newcommand{\Lcwa}{\mathcal{L}_{\cwa}}
\newcommand{\Ldet}{\mathcal{L}_{\det}}
\newcommand{\Lfix}{\mathcal{L}_{\fix}}
\newcommand{\Lq}{\mathcal{L}_{\mathsf{Q}}}

\newcommand{\const}{\mathbf{Const}}
\newcommand{\rels}{\mathbf{Rel}}
\newcommand{\theories}{\mathbf{T}}
\newcommand{\queries}{\mathbf{Q}}

\newcommand{\arity}{\mathit{arity}}

\newcommand{\semantics}{\mathsf{MOD}}

\newcommand{\answer}[2]{\llbracket #1 \rrbracket_{#2}}

\newcommand{\mcI}{\mathcal{I}}
\newcommand{\mcJ}{\mathcal{J}}
\newcommand{\mcF}{\mathcal{F}}

\newcommand{\dlkb}{\mathcal{O}}

\newcommand{\conceptnames}{\mathsf{N_C}}
\newcommand{\simpleconcepts}{\mathsf{N_C^{+}}}
\newcommand{\rolenames}{\mathsf{N_R}}
\newcommand{\roles}{\mathsf{{N}_R^{+}}}

\newcommand{\alnull}{\aleph_0}
\newcommand{\nat}{\mathbb{N}}
\newcommand{\natstar}{\mathbb{N}^{*}}

\newcommand{\tiles}{\mathsf{Tiles}}

\newcommand{\func}{\mathsf{func}}

\newcommand{\UCQ}{\mathcal{UCQ}}
\newcommand{\CQ}{\mathcal{CQ}}
\newcommand{\AQ}{\mathcal{AQ}}
\newcommand{\IQ}{\mathcal{IQ}}

\newcommand{\parTELIbot}{\mathsf{T}{:}\,\ELIbot}
\newcommand{\parTALCHOI}{\mathsf{T}{:}\,\ALCHOI}
\newcommand{\parTALCO}{\mathsf{T}{:}\,\ALCO}
\newcommand{\parTALCOI}{\mathsf{T}{:}\,\ALCOI}
\newcommand{\parTALCHOIF}{\mathsf{T}{:}\,\ALCHOIF}
\newcommand{\parTALCHIF}{\mathsf{T}{:}\,\ALCHIF}
\newcommand{\parTLITEBOOL}{\mathsf{T}{:}\,\DLLITEBOOL}

\newcommand{\parTLITEOBOOL}{\mathsf{T}{:}\,\DLLITEOBOOL}
\newcommand{\parTLITEFUNC}{\mathsf{T}{:}\,\DLLITEFUNC}

\newcommand{\parTL}{\mathsf{T}{:}\,\Lth}

\newcommand{\parCCQ}{\mathsf{C}{:}\,\CQ}
\newcommand{\parCL}{\mathsf{C}{:}\,\mathcal{L}_{\cwa}}
\newcommand{\parCAQ}{\mathsf{C}{:}\,\AQ}
\newcommand{\parCIQ}{\mathsf{C}{:}\,\IQ}
\newcommand{\parCempty}{\mathsf{C}{:}\,\emptyset}

\newcommand{\parDCQ}{\mathsf{D}{:}\,\CQ}
\newcommand{\parDL}{\mathsf{D}{:}\,\mathcal{L}_{\det}}

\newcommand{\parDempty}{\mathsf{D}{:}\,\emptyset}
\newcommand{\parDany}{\mathsf{D}{:}\,\queries}

\newcommand{\parFL}{\mathsf{F}{:}\,\mathcal{L}_{\fix}}
\newcommand{\parFAQ}{\mathsf{F}{:}\,\AQ}
\newcommand{\parFempty}{\mathsf{F}{:}\,\emptyset}

\newcommand{\parQCQ}{\mathsf{Q}{:}\,\CQ}
\newcommand{\parQL}{\mathsf{Q}{:}\,\mathcal{L}_{\mathsf{Q}}}
\newcommand{\parQAQ}{\mathsf{Q}{:}\,\AQ}
\newcommand{\parQIQ}{\mathsf{Q}{:}\,\IQ}

\newcommand{\recognition}{\mathsf{FOCUS}}
\newcommand{\entailment}{\mathsf{ENTAILMENT}}
\newcommand{\consistency}{\mathsf{CONSISTENCY}}
\newcommand{\emptiness}{\mathsf{EMPTINESS}\xspace} 
\newcommand{\mixedsat}{\mathsf{MIXED\textsf{-}SAT}}
\newcommand{\mixedunsat}{\mathsf{MIXED\textsf{-}UNSAT}}

\newcommand{\mixedqa}{\mathsf{MIXED\textsf{-}OMQA}}
\newcommand{\qempty}{\mathsf{NULLABILITY}}

\newcommand{\problemdef}[3]{
  \smallskip
\begin{center}
\begin{tabular}{|r p{19.9em} |}
  \hline
\multicolumn{2}{|l|}{\rule{0pt}{2.6ex}\textbf{#1}} \\[1ex]
\textit{Input:}     & #2\\[1ex]
\textit{Question:}  & #3 \\[.5ex]
\hline
\end{tabular}
\end{center}
\smallskip
}

\newcommand{\problemdefprom}[4]{
  \smallskip
  \begin{center}
\begin{tabular}{|r p{19.9em} |}
  \hline 
\multicolumn{2}{|l|}{\rule{0pt}{2.6ex}\textbf{#1}} \\[1ex]
\textit{Input:}   & #2\\[1ex]
                  &  \vspace*{-10pt}  #3 \vspace*{-10pt}\\[1ex]
  \textit{Question:}  & #4
  \\ 
\hline
\end{tabular}
\end{center}
\smallskip
}


%% file: 1-introduction.tex
\section{Introduction} 
In the design of data-centric systems, coming up with the right data
organization (in terms of database schemas, integrity
constraints, conceptual models, etc.) is of paramount importance. If
well-chosen, it can make the implementation of the remaining \linebreak
functionality more evident, as it binds the developers to one shared
and unambiguous view of the data to be managed by the target
system. Unfortunately, coming up with the right data organization
remains challenging and time-consuming, despite the many techniques
and tools that are available to aid the design of data-centric
systems. In addition, modern systems face further challenges, like
incompleteness of information, or the need for interoperability with
multiple other
systems~\cite{DBLP:journals/dagstuhl-manifestos/AbiteboulABBCD018}. 
 


We propose a novel way to exploit domain knowledge
captured in ontologies in the design of data-centric systems.
Ontologies, understood here as logical theories
expressing domain knowledge, provide a shared understanding of
the domain to different users and applications; justified by expected
reusability, considerable resources have been invested in constructing
high-quality ontologies for many domains~\cite{Horrocks08}.   In data management they
have already proved to be a powerful tool.  Successful applications
include data integration and querying incomplete data sources, where
they are used on-line during the system operation to infer additional
facts from incomplete data, and to provide a unified view of
heterogenous data sources~\cite{PoggiLCGLR08,BenediktGK18,XiaoCKLPRZ18}.
%
%
%
Here, we would like to use an existing ontology to produce, quickly
and with moderate effort, data-centric systems on demand.
To achieve this goal, two key challenges need to be overcome.
\begin{description}
\item[Specificity.] Ontologies are typically
  broad, containing many terms irrelevant for the intended
  application.  We need methods to restrict their scope to  
  obtain more manageable conceptualizations. 
\item[Data completeness.] Ontologies have an  \emph{open-world} semantics that
  treats data as incomplete, while every meaningful data-centric application
  will call for some completeness assumptions on (parts of) the data. 
\end{description} 
%
As a response to these challenges, we propose a process called \emph{focusing} that allows
us to trim away irrelevant information, and establish completeness
assumptions.  The goal of this process is to find so-called
\emph{focusing solutions}.  Syntactically, a focusing solution
provides a schema for a database and specifies how its instances are
enriched with the knowledge from the ontology, by prescribing which
`parts' of the ontology are relevant and which are complete.
Semantically, focusing solutions define a set of \emph{intended
models} for these knowledge-enriched databases: those that give the
expected answers for the relevant queries.

Our main contribution is formalizing focusing solutions.
The notion is independent from the ontology language and gives several
options for specifying the scope of the system. 
We 
 identify key
computational problems relevant for obtaining and using focusing
solutions. As an advanced proof of concept, we instantiate our general
notions by considering a few choices of ontology languages and
scope specifications. For these combinations, we study the
decidability and complexity of the introduced computational problems.
As a by-product, we isolate (and solve) variants of classical
reasoning tasks that seem to be interesting in their own right.

%% file: 3-focusing.tex
\section{General Framework}

We now discuss the key features and design choices in the
framework we propose. We begin from a motivating use case.

\input{2-example.tex}

\subsection{What should focusing be?} 

The starting point is an ontology, expressed as a theory in some
logical formalism over a relational signature.  Based on some
specification of the intended scope of the system, we need to decide
which predicates are to be supported by the system and how 
the database will interact with the ontology at run time. Let us
assume that we only allow ourselves two make two decisions about each
predicate: whether it will evolve during the run of the system and so
should be stored in the database (dynamic vs static predicates), and
whether it should be viewed as a complete representation of the data
it represents (closed vs open predicates). The focusing solutions are
designed to support these decisions, as well as the specification of
the scope of the system. We now informally describe the four
components of focusing solutions.

\smallskip

A bare-bones focusing solution is a \textbf{database schema}, collecting all
dynamic predicates. Such system allows storing relevant data and
updating it as the reality evolves, but neither gives any
completeness guarantees, nor improves the specificity.

\smallskip

We address the \textbf{data completeness} issue by declaring some
predicates as \emph{closed}. Semantically, these declarations trim the
set of represented models, keeping only those that agree with the
current database instance on the closed predicates. In particular, if
we decide to close all dynamic predicates, we end up with a standard
database, with the ontology acting as a set of integrity
constraints.

The actual intention behind declaring a predicate as closed, is 
to commit the system to keep the stored content of the predicate
identical to its real-world interpretation. Thus, closed dynamic predicates
should be used to represent changing aspects of reality that are fully
observable to the system. In our motivating example, these could
include the precise list of districts for which an evacuation was
ordered, the open shelters, assignments of personal to tasks,
shelters, etc. Some of these aspects may be fully observable because
they are actually controlled by the system (maybe the evacuation orders are
issued by the system itself), for others we might rely on updates from
some other trusted system.

\smallskip

We address the \textbf{specificity} issue by declaring some predicates as
\emph{fixed}. When we fix a predicate, we require its extension to be
determined by the ontology alone. That is, in an intended model, a
fixed predicate will contain a tuple of constants only if the ontology
alone entails the corresponding atom. When we fix a predicate, it
effectively becomes static and closed: even if it is stored in the
database, it has only one allowed extension.

If a predicate is \emph{not} populated by the ontology (which is quite
common since most ontologies focus on \emph{terminological} rather
than \emph{assertional} knowledge), fixing it enforces its extension
to be empty. By fixing irrelevant predicates, we avoid reasoning about
them and make the ontology more specific for the current situation.
This reflects the intuition that the more specific our knowledge of
the situation is, the more inferences we can draw from the ontology
about our situation.

Fixed predicates will typically include aspects of reality that are captured in
the ontology but irrelevant to our application, like the specifics of
other types of emergencies not related to the current one. We can also
fix predicates that are relevant, but correspond to immutable aspects
of reality, and their extension is uniquely (and accurately)
determined by the ontology. For instance, all the districts of a city,
or all hospitals in a district.

A fixed predicate can in principle act as an additional integrity
constraint. This happens if declaring the predicate as fixed discards
all intended models for some instance, thus making the instance
\emph{inconsistent}. This is undesired, and will be explicitly
forbidden in the definition of focusing solution.  



\smallskip

Our focusing solutions already provide a database schema and an
interface to the ontology that addresses the specificity and data
completeness issues, but so far we have been assuming that the
designer makes all the choices, as appropriate for the intended scope
of the system. In order to support these choices using automated
reasoning, or at least check that they are correct, we need the
designer to provide a formal \textbf{specification of the scope}. We
propose to specify the scope in terms of queries that will be posed
when the system is running. We shall have three families of such queries.

\emph{Determined queries} are the ones for which we want a guarantee
that the answers entailed by the data and the ontology are complete
with respect to all possible models. Equivalently, the answers do not
depend on the concrete model, as long as it is compatible with the
database contents and the ontology. Declaring a query as determined
should be viewed as a demand of the designer:  this query needs to be
determined, how do we guarantee this?

\emph{Fixed queries} generalize fixed predicates: complete answers to
these queries are entailed by the ontology alone. Declaring a query
as fixed is a decision rather than a demand: We freeze the answers as
the ones entailed by the ontology alone, and models yielding more answers
should be discarded. The more assumptions of this kind, the easier it
is to make other queries determined.

\emph{Closed queries} generalize closed (dynamic) predicates. For
these queries, we are making the assumption that complete answers can
be obtained from the data alone, by directly evaluating the
query. Declaring a query as closed is a promise made by the designer:
I am prepared to maintain the data in such a way that this query is
closed. Again, the more assumptions of this kind, the easier it is to
make other queries determined.

Note that, in fact, all three families of queries can be viewed as
completeness assertions: closed queries talk about completeness of the
data, fixed queries talk about completeness of the ontology, and
determined queries talk about completeness of the combination of
both. Allowing fixed and closed queries, rather than just predicates,
gives a bit more flexibility to the designer. For example, it might
not be reasonable to assume complete knowledge about all buildings in
the city, or all hospitals in the country, but maintaining up-to-date
information about these buildings in the city that are hospitals is a
perfectly reasonable requirement.

\smallskip

The \emph{focusing} problem is to find a \emph{focusing solution} that 
guarantees that certain queries are determined,
assuming that certain other queries are closed or fixed.
These solutions are in general not unique and there are
many trade-offs involved. For example, the more predicates are closed,
the easier it is for a query to be determined, but we must pay the
maintenance costs for each predicate we close. It seems desirable to
have as few closed predicates as possible, provided we can guarantee
that the suitable queries are determined. On the other hand, if
suitable queries are already determined, it may be desirable
to fix as many predicates as is possible without making any instances
inconsistent. 
In the following subsection we formalize the outcome of this
discussion.

\subsection{What is a focusing solution?} \label{sec:what-focusing}

We shall keep our notions independent from the
specific forma\-lisms used to express ontologies and queries.

 We assume an infinite set $\const$ of \emph{constants}
and  an infinite set $\rels$ of \emph{relation symbols}.
 Each relation symbol $r \in \rels$ has a non-negative integer
  \emph{arity}, denoted by $\arity(r)$.
 An \emph{atom} is an expression of the form $r(c_1,\ldots,c_n)$,
  where $r\in \rels$, $\{c_1,\ldots,c_n\}\subseteq\const$, and $n=\arity(r)$.
 A \emph{(database) instance} $\mcI$ is any set of atoms.
 A \emph{signature} $\Sigma$ is any set of relation symbols.
 An instance $\mcI$ is \emph{over} a signature $\Sigma$, if
  $r(\vec{t})\in \mcI$ implies $r\in\Sigma$.
 The \emph{active domain} of $\Ii$, denoted by $\adom{\Ii}$, is
  the set of all constants in the atoms of $\Ii$.
     
 We assume an infinite set $\theories$ of \emph{theories}.
 Each theory $\varphi\in \theories$ is associated with a set of
  instances that are called \emph{models} of $\varphi$.
 We write $\mcI\models \varphi$ if $\mcI$ is a model of
  $\varphi$.
 
 We assume an infinite set $\queries$ of \emph{queries}.
 Each query $q\in \queries$ has a non-negative integer arity,
  denoted $\arity(q)$.
 Each query $q\in\queries$ is associated with a function
  $\answer{\cdot}{q}$ that maps every instance $\mcI$ to an $n$-ary
  relation $\answer{\mcI}{q}\subseteq \const^{n}$, where
  $n=\arity(q)$.
   Queries of arity 0 are called \emph{Boolean}; their associated
    functions map instances to subsets of $\const^{0} =
    \{\varepsilon\}$, where $\varepsilon$ is the empty tuple. A
    Boolean query \emph{holds} in an instance $\mcI$, if
    $\answer{\mcI}{q} = \{\varepsilon\}$.
 We need the notion of \emph{certain answers}. Given an $n$-ary
  query $q\in \queries$, a theory $\varphi$ and an instance $\mcI$, we let
  $\answer{\varphi,\mcI}{q}$ denote the set of $n$-tuples $\vec{u}\in \const^{n}$
  satisfying the following implication: if
  $\mcJ\models \varphi$ and $\mcI\subseteq \mcJ$, then $\vec{u}\in \answer{\mcJ}{q}$.

A \emph{focusing configuration} is a database schema together with \linebreak
three sets of queries representing completeness assertions about the data
and about the theory, and determinacy assertions.

\begin{definition}[Focusing configuration]
  \label{def:focusing-configuration}
  A (focusing) configuration is a tuple
  $\mcF=(\Sigma,\Qcwa,\Qfix,\Qdet)$, where $\Sigma\subseteq \rels$ is
  a signature, and $\Qcwa,\Qfix,\Qdet\subseteq \queries$ are sets of
  queries.  An instance $\mcI$ is \emph{legal} for $\mcF$ in case it
  is over $\Sigma$.
\end{definition}

Let us explain how the four ingredients of a focusing configuration
$\mcF=(\Sigma,\Qcwa,\Qfix,\Qdet)$ work. First, the signature $\Sigma$
is a database schema that determines database instances of
interest: an instance is \emph{legal} for $\Ff$ it is over the
signature $\Sigma$. 


The queries from $\Qcwa$ specify completeness assertions
about data, effectively restricting the set of models represented by a
legal instance as follows.

\begin{definition}[Query-based Completeness] \label{def:closed-queries}
  Given a theory $\varphi$, an instance $\mcI$, and a set $Q$ of
  queries, we let $\rest(\varphi,\mcI,Q)$ be the set of all instances
  $\mcJ$ such that
  \begin{enumerate}
  \item $\mcI\subseteq\mcJ$,
  \item $\mcJ\models \varphi$, and
  \item $\answer{\mcI}{q}=\answer{\mcJ}{q}$ for all $q\in Q$.
  \end{enumerate}
\end{definition}

Intuitively, $\rest(\varphi,\mcI,Q)$ contains exactly the models of
$\varphi$ and $\mcI$ that provide no new information about the queries
in $Q$ compared to the information given by $\mcI$ alone. 

The queries in $\Qfix$ specify completeness assertions about the theory,
further restricting the set of represented models.

\begin{definition}[Query-based Fixing]
  For a theory $\varphi$, an instance $\mcI$, and a set $Q$ of
  queries, we let $\fix(\varphi,\mcI,Q)$ be the set of all instances
  $\mcJ$ such that
  \begin{enumerate}
  \item $\mcI\subseteq\mcJ$,
  \item $\mcJ\models \varphi$, and
  \item $\answer{\mcJ}{q}=\answer{\varphi,\emptyset}{q}$ for all
    $q\in Q$.
  \end{enumerate}
\end{definition}

Intuitively, $\fix(\varphi,\mcI,Q)$ contains exactly the models of
$\varphi$ and $\mcI$ that provide no new information about the queries
in $Q$ compared to the information given by 
$\varphi$ alone.


Thus, a configuration $\mcF=(\Sigma,\Qcwa,\Qfix,\Qdet)$ tells us to
consider only database instances $\mcI$ over $\Sigma$ as legal. For
each such instance $\mcI$, we are to assume that the information
retrieved from $\mcI$ alone by the queries in $\Qcwa$ is complete, and 
we are to restrict our attention to models where 
the answers to queries in $\Qfix$ are frozen to the facts inferred
from the initial theory $\varphi$ (without any data). These are the
intended models represented by a concrete legal database instance
$\mcI$.

\begin{definition}[Intended models]
  For  a theory $\varphi $, a configuration
  $\mcF= (\Sigma,\Qcwa,\Qfix,\Qdet)$, and an instance
  $\mcI$ over $\Sigma$, let  \[\semantics(\varphi,\mcF,\mcI)=\rest(\varphi,\mcI,\Qcwa)\cap
  \fix(\varphi,\mcI,\Qfix)\,.\]
\end{definition}

We are now ready to provide the definition of focusing, which makes
the role of $\Qdet$ precise: we shall demand that freezing answers to
the queries from $\Qfix$ does not affect consistency of instances and
that the answers to the queries in $\Qdet$ coincide over all intended
models.

\begin{definition}[Focusing]\label{def:focusing} 
  A \emph{focusing solution} for a theory $\varphi $ is a
  configuration $\mcF=(\Sigma,\Qcwa,\Qfix,\Qdet)$ such that the
  next two conditions are obeyed \emph{for all} instances $\mcI$~over~$\Sigma$:
  \begin{enumerate}
  \item if $\rest(\varphi,\mcI,\Qcwa)\neq \emptyset$, then
    $\semantics(\varphi,\mcF,\mcI) \neq \emptyset$;
  \item for all
    $\mcJ_1,\mcJ_2\in \semantics(\varphi,\mcF,\mcI)$, and all $q\in \Qdet$,  we have
    $\answer{\mcJ_1}{q}=\answer{\mcJ_2}{q}$.
  \end{enumerate}
 \end{definition}

Let us see how this definition captures the scenario discussed in the
previous subsection. Suppose that the designer specifies a set $\Qdet$
determined queries that need to be guaranteed as well as sets
$\Qcwa^u$ of queries that are promised to be closed, and $\Qfix^u$ of
queries are chosen to be interpreted as fixed.  We now want to find a
correct focusing solution of the form
 \[ \mcF=\left(\Sigma, \, \Qcwa, \Qfix,\, \Qdet\right)\] with $\Qcwa^u
\subseteq \Qcwa$ and $\Qfix^u \subseteq \Qfix$.  There may be many
such focusing solutions, and they are all good in the sense that they
guarantee correct answers to $\Qdet$.  This leaves some space to
accommodate additional preferences of the designer; we discuss it in
the following subsection.

\subsection{Getting and using focusing solutions}

We shall now identify key reasoning problems crucial in obtaining and
using focusing solutions. For each problem the input includes a
focusing configuration $\mcF$; some problems additionally input
theories, database instances, and queries. To be able to speak about
concrete formalisms, we parameterize the problems by query and
ontology languages. We write $\parTL$ to indicate that the language
used for expressing theories is $\Lth \subseteq
\theories$. Similarly, we use $\parCL$, $\parFL$, $\parDL$, and
$\parQL$ for $\Lcwa, \Lfix, \Ldet, \Lq\subseteq \queries$.

The main problem is recognizing focusing solutions among
focusing configurations. 

\problemdefprom{$\recognition(\parTL,\parCL,\parFL,\parDL)$}{A pair
  $(\varphi,\mcF)$ with  $\mcF=
    (\Sigma,\Qcwa,\Qfix,\Qdet)$, and}{
    \begin{enumerate}[-]
  \item $\varphi\,{\in}\, \Lth$,
    \item $ \Sigma\subseteq \rels$, 
      \item $\Qcwa \subseteq \Lq$, 
        $\Qfix\subseteq \Lfix$,
        $\Qdet \subseteq \Ldet$.
      \end{enumerate}
      \saveproblem
}{Is $\mcF$ a focusing solution to $\varphi$?}

Thus, in the focusing problem a candidate focusing configuration is
given, and the task is to decide if it is a focusing solution. In the
scenario discussed previously, the input consists of closed queries
$\Qcwa^u$, fixed queries $\Qfix^u$, and determined queries $\Qdet$,
and the output is a focusing solution $\left(\Sigma, \,
\Qcwa, \Qfix,\, \Qdet\right)$ with $\Qcwa^u \subseteq \Qcwa$ and
$\Qfix^u \subseteq \Qfix$.
If we consider for $\Qcwa$ and $\Qfix$ a query language that gives us
a finite number of candidates, like atomic queries (which covers the
basic scenario with closed and fixed predicates), the recognition
problem can be used directly in the search for such a solution by
applying exhaustive search. This search can be guided by some
preferences of the designer. For example, a basic strategy could be to
minimize the set of closed queries, and then maximize the set of fixed
ones. More sophisticated strategies could involve a specified order in
which the designer prefers to close predicates, reflecting, for
instance, the cost of maintaining them (size statistics, availability,
acquisition costs, etc). One could also consider semi-automated
approaches, like a dialog approach where successive solutions are
proposed to the designer, who adjusts the specification, or even the
ontology, and accepts some suggested choices while rejecting others,
thus converging to a satisfactory focusing solution. We leave
investigating such strategies to future research.



An additional criterion that might be useful in the search for 
suitable focusing solutions is the existence of consistent
database instances. This is embodied in the following decision
problem.\linebreak Here $\mathcal{P}$ stands for a collection of parameters.
\problemdef{$\emptiness(\mathcal{P})$}{A triple $(\varphi,\mcF)$,
  where $(\varphi,\mcF)\in\recognition(\mathcal{P})$}
{Is $\semantics(\varphi,\mcF,\mcI) = \emptyset $ for each $\mcI$ legal for $\mcF$?}
Obviously a single consistent database instance does not make a
focusing solution very useful, but the criterion can help eliminate
some utterly useless solutions. 

Assuming that a satisfactory focusing solution is found, how do we use
it? Two reasoning tasks are crucial in the operation of the system
resulting from focusing. First, whenever a tuple is inserted, deleted,
or updated, the system needs to check that the resulting database
instance is still consistent. This makes the feasibility of the
following problem essential.
\problemdefprom{$\consistency(\mathcal{P})$}{A triple $(\varphi,\mcF,\mcI)$, where}{
  \begin{enumerate}[-]
  \item $\mcI$ is legal for $\mcF$,
  \item $(\varphi,\mcF)\in\recognition(\mathcal{P})$.
  \end{enumerate}
  \saveproblem
}{$\semantics(\varphi,\mcF,\mcI)\neq \emptyset $ ? }

Finally, the system is there to be queried. This makes the following
entailment problem relevant. 
\problemdefprom{$\entailment(\parQL, \mathcal{P})$}{A tuple $(\varphi,\mcF,\mcI,q)$, where}{  
  \begin{enumerate}[-]
    \item $\mcI$ is legal for $\mcF$,
    \item $q\in \Lq$ is Boolean,
    \item $(\varphi,\mcF)\in\recognition(\mathcal{P})$.
    \end{enumerate}
    \saveproblem
}{Is $q$ true in all $\mcJ\in \semantics(\varphi,\mcF,\mcI)$?}
Note that $\consistency$ is a special case of non-$\entailment$, but
not conversely, so complexity of the two problems might differ.


%% file: 2-example.tex
\subsection{Emergency response in a smart city} 
City authorities rely on an ontology to quickly build
situation-specific applications supporting response to emergency 
events. \linebreak The ontology contains, among others: 
\begin{inparaenum}
\item\label{city} Data about the city, e.g., districts, population, 
  facilities such as hospitals, schools, and sport centers, together
  with associated details like their capacity, size, and existing services;
\item\label{risks} 
The city's risk assessment data, detailing possible emergencies;
\item\label{emergency} 
Knowledge about public health emergencies, e.g., types of emergencies include 
disease outbreaks, radiation emergencies,
and disasters and weather emergencies; the latter
 include extreme heat, floods, hurricanes, and wild fires;
\item
Emergency response knowledge, like types of responders (paramedics,
firefighters, police officers, etc.); 
\item\label{gral} 
General knowledge about relevant topics like weather, buildings, or cities.
\end{inparaenum}

Having such an ontology for this purpose is possible and
realistic. In fact, many cities now compile and store  data  concerning (\ref{city})
and (\ref{risks}), and there are vast repositories of
online resources and readily available ontologies for
(\ref{emergency})--(\ref{gral}); for example, see \cite{smartcity}.
Moreover, maintaining such an ontology can be part of the routine 
preparation measures carried out in emergency response departments
when there are no emergencies.

When a disaster happens, an automated focusing engine is fed with
relevant parameters, to obtain focusing solutions that can help
quickly build tools specific to this kind of disaster. For example, in
many emergencies, so-called `community assessment' questionnaires are
used to gather critical public health data such as availability of
drinking water and electricity in households. Existing guidelines
suggest to first design the database that will be used to gather and
analyze the data, and to guide the questionnaire design with it. Using
an ontology like the one we have described, one could instead propose
the database tables automatically. In the event of a disaster like a
flood, one could quickly create an application for storing, for
example, the current alert level in the different districts of the
city and the complete list of operating shelters together with
currently available resources, while disregarding any knowledge in the
ontology related to other possible disasters, such as wild fires or
extreme heat.

%% file: 5-concrete-problems.tex
\section{Concrete problems} 
\label{sec:instantiation}

To illustrate some concrete settings that may be useful, in this
section we instantiate the reasoning problems with some selected
languages.  We discuss how the problems can be solved in these
concrete cases, and provide complexity results for them.

\subsection{Ontology and Query Languages}


Here we recall the concrete formalisms for expressing theories and
queries that we study. As ontology languages we look at description
logics (DLs), and as query languages we consider instance queries, atomic queries, 
and conjunctive queries (CQs).


\subsubsection*{Description Logics (DLs)}\label{sec:descr-logics-dls}
DLs is a family of logics, specifically tailored
for writing ontologies~\cite{DLNewBook}. Most DLs are fragments of first-order logic,
which allow only unary and binary relation symbols. This and other
restrictions allow DLs to be equipped with a special syntax, which
allows to write formulas in a more concise way. We now introduce 
the main DL of this paper, called $\ALCHOIF$.

Let $\conceptnames\subseteq\rels$ be a countably infinite set of unary
relation symbols, called \emph{concept names}, and let
$\rolenames\subseteq\rels$ be a countably infinite set of binary
relation symbols, called \emph{role names}. If $r\in \rolenames$, then
$r$ and the expression $r^{-}$ are \emph{roles} ($r^{-}$ is also
called the \emph{inverse of} $r$). For a role $r^{-}$, we let
$(r^{-})^{-}=r$. We let
$\roles = \rolenames\cup \{r^-\mid r\in\rolenames\}$.  The set of
$\ALCHOIF$ \emph{concepts} is defined inductively as follows:
\begin{inparaenum}[(a)]
\item every concept name $A\in \conceptnames$ is a concept;
\item the  expressions $\top$ and $\bot$ are concepts;
\item the expression $\{c\}$, where  $c\in \const$, is also a concept (called \emph{nominal});
\item if $C,D$ are concepts, and $p$ is a role, then $\neg C$,
  $C\sqcap D$, $C\sqcup D$, $\exists p.C$ and $\forall p.C$ are also
  concepts. 
\end{inparaenum} 
A \emph{concept inclusion} is an expression of the form $C\ISA D$,
where $C,D$ are concepts. A \emph{role inclusion} is an expression of
the form $r\ISA p$, where $r,p$ are roles. A \emph{functionality
  assertion} is an expression of the form $\func(p)$, where $p$ is a
role. An ($\ALCHOIF$) \emph{ontology} $\Oo$ is a finite set of concept
inclusions, role inclusions, and functionality assertions. The DL that
is obtained by disallowing functionality assertions, role inclusions,
or nominals is indicated by, respectively, dropping `$\mathcal{F}$',
`$\mathcal{H}$', or `$\mathcal{O}$' from its name.
An ontology is in the DL $\ELIbot$, if it contains only
concept inclusions, and they are built using only the constructs
$\sqcap$ and $\exists$ (i.e., nominals, $\sqcup$, and $\forall$, as
well as role inclusions and functionality assertions, are forbidden).
We use $\conceptnames(\Oo)$ and $\rolenames(\Oo)$ to denote the sets
of concept names and role names that appear in $\Oo$, respectively.  We let
$\roles(\Oo)=\rolenames(\Oo)\cup \{r^-\mid r\in \rolenames(\Oo) \}$.

The semantics to ontologies is given using instances, defined in
Section~\ref{sec:what-focusing}. Assume an instance $\mcI$ and an
ontology $\Oo$. We define a function $\cdot^{\mcI}$ that maps every
concept $C$ to a set $C^{\mcI}\subseteq \const$, and every role $p$ to
a relation $p^{\mcI}\subseteq \const\times \const$. For a concept name
$A$, and a role name $r$, we let $A^\mcI = \{c\mid A(c)\in \mcI\}$ and
$r^\mcI = \{(c,d)\mid r(c,d)\in \mcI\}$. The function $\cdot^{\mcI}$
is then extended to the remaining concepts and roles as follows:
\begin{gather*}
  \top^{\mcI} = \adom{\Ii},\,\qquad
  \bot^{\mcI} = \emptyset\,, \qquad
  \{c\}^{\mcI}  =\{c\}\,, \\
  (C\sqcup D)^{\mcI}   = C^{\mcI} \cup D^{\mcI} \,, \qquad
  (C\sqcap D)^{\mcI}   = C^{\mcI} \cap D^{\mcI}\,, \\
  \begin{align*}
      (\lnot C)^{\mcI}   & = \adom{\Ii} \setminus C^{\mcI}\,,\\
      (\exists p.C)^{\mcI}  &= \{c\mid \exists d: (c,d)\in p^\mcI \mbox{ and } d\in C^{\mcI}\} \,,\\
    (\forall p.C)^{\mcI}  &= \{c\mid \forall  d: (c,d)\in p^\mcI \mbox{ implies }d\in C^{\mcI}\} \,,\\
    (r^{-})^{\mcI}  &= \{(c,d)\mid (d,c)\in r^{\mcI}\} \,.
\end{align*}
\end{gather*}
  We say $\mcI$ is a \emph{model} of $\Oo$, if the following are satisfied:
  \begin{inparaenum}[(1)]
  \item $\alpha^{\mcI}\subseteq \beta^{\mcI}$ for all concept and role
    inclusions $\alpha\ISA \beta\in \Oo$, and
  \item   $(c,d_1)\in p^{\mcI}$ and $(c,d_2)\in p^{\mcI}$ imply
    $d_1=d_2$, for all $c\in \const$ and $\func(p)\in \Oo$. 
  \end{inparaenum}
We note that by defining the semantics using database instances,
instead of general first-order structures (or \emph{interpretations},
as they are called in DLs), we are effectively interpreting ontologies
under the \emph{Standard Name Assumption (SNA)}. That is, the domain
of interpretation is always the set $\const$, and the interpretation
of constants is given by the identity function. However, the \emph{active
domain} of the instance, as defined in Section \ref{sec:what-focusing},
may well be a proper subset of $\const$. 

In order to simplify presentation, when providing upper bounds we
can concentrate wlog. on ontologies in \emph{normal form}, which
is defined as following.  A \emph{simple concept} $B$ is any concept
in $\const\cup\{\top,\bot\}\cup\{\{c\}\mid c\in\const\}$. We use
$\simpleconcepts(\Oo)$ to denote the set of simple concepts that
appear in $\Oo$.  An ontology $\Oo$ is in \emph{normal form} if all
its statements are of one of the following forms:
\begin{gather*}
  B_1\sqcap \ldots \sqcap B_{k-1} \ISA B_{k}\sqcup\ldots \sqcup B_m\,,\\
  B_1\ISA \exists r. B_2\,, \qquad  B_1\ISA \forall r. B_2\,, \qquad  r\ISA s\,, \qquad \func(r)\,,
\end{gather*}
where $B_1,\ldots,B_m$ are simple concepts, $k > 1$, $m\geq k$, and $r,s$ are roles. It
is well known that any ontology $\Oo$ can be transformed into an
ontology $\Oo'$ in normal form such that $\Oo$ and $\Oo'$ have the
same models up to the signature of $\Oo$.

We will also study  DLs whose definition is based on the above normal form. In particular,  $\DLLITEBOOL$ ontologies are ontologies in normal form that additionally satisfy the following:
\begin{enumerate}[(i)]
\item for all  $B_1\ISA \exists r.B_2$, we have  $B_2=\top$, 
\item for all $B_1\ISA \forall r.B_2$, we have $B_1=\top$, and
 \item  functional roles have no subroles: $\func(s)\in \dlkb$ and
$r \ISA s \in \Oo$ imply $r=s$. 
\end{enumerate}
The DL $\DLLITEFUNC$ is obtained by applying the above restrictions, but
additionally prohibiting nominals, as well as $\sqcup$ and $\sqcap$;
that is, we only have concept inclusions of the form $B_1\ISA B_2$
and the two forms (i) and (ii) shown above. Observe that
\[\DLLITEFUNC\subseteq \DLLITEBOOL\subseteq \ALCHOIF\,.\]

   \subsubsection*{Queries} 

We let $\CQ$ be the class of conjunctive queries (primitive positive
first-order formulas) over the signature $\conceptnames
\cup\rolenames$, with the usual semantics. Occasionally we talk about
$\UCQ$, the class of unions of conjunctive queries (UCQs),
corresponding to positive existential first-order formulas.  We also
consider the class $\AQ \subseteq \CQ$ of atomic queries, and the
class $\IQ \subseteq \AQ$ of instance queries, that is, atomic queries
over concepts. If $Q\subseteq \AQ$, we can view $Q$ as a set of
predicates.


\input{51-recognition}

\input{53-entailment}

\input{54-emptiness}

%% file: 51-recognition.tex
\subsection{Recognizing focusing solutions}\label{sec:recogn}

In this section we deal with recognizing focusing solutions for
$\ALCHOI$ ontologies. The two conditions in the definition of focusing
solutions (Definition~\ref{def:focusing}) are closely related to
natural variants of two classical problems in description logics.

The first condition of Definition~\ref{def:focusing} boils down to a
variant of the classical query entailment problem that only considers
models where selected predicates have finite extensions (the
remaining predicates may have a finite or an infinite extension). 

\problemdefprom{$\mixedqa(\parTL,\parQL)$}{A tuple
  $(\varphi,\Sigma,\mcI, q)$ where
  }{
  \begin{enumerate}[-]
    \item $\varphi \in  \Lth$,
    \item $ \Sigma \subseteq \rels$,
    \item $\mcI$ is an instance,
      \item $q \in\Lq$ is Boolean.
      \end{enumerate}
      \saveproblem
}{Is $q$ satisfied in every model $\mcJ$ of $\varphi$ such that
 \begin{asparaenum}[-]
 \item $\mcI\subseteq \mcJ$, and
 \item $R^{\mcJ}$ is finite for all $R\in \Sigma$?
 \end{asparaenum}
  \vspace*{-10pt}}
This mixed variant generalizes the usual finite and unrestricted
variants of OMQA. Recall that a logic has finite controllability if,
for each theory expressed in this logic, non-entailed queries have
finite counter-models. For logics enjoying finite controllability, the
three variants of OMQA coincide. This is the case for $\ALCHOI$, which makes
our problem $\twoexptime$-complete \cite{CalvaneseEO09}.

In the second condition of Definition~\ref{def:focusing}, the
computational crux of the matter is the following problem.
\problemdefprom{$\qempty(\parTL,\parCL,\parQL)$}{A tuple
  $(\varphi,\Sigma, \Qcwa, q)$ where
}{
  \begin{enumerate}[-]
    \item $\varphi \in  \Lth$,
    \item $ \Qcwa \subseteq \Lcwa$,
      \item $q \in\Lq$.
      \end{enumerate}
      \saveproblem
}{
  Do all $\mcI$ over  $\Sigma$ with $\rest(\varphi,\mcI,
  \Qcwa)\neq \emptyset$
  admit $\mcJ \in \rest(\varphi,\mcI, \Qcwa)$ with
  $\answer{\mcJ}{q} = \emptyset$? 
}
%
The nullability problem is closely related to the query emptiness
problem, which is known to be $\nexptime$-complete for atomic queries, 
$\ALCI$ ontologies, and no closed predicates
\cite{DBLP:journals/jair/BaaderBL16}. We show that nullability of
atomic queries for $\ALCHOI$ with closed instance queries (i.e.,
closed concepts) is in $\conexptime^\np$. Allowing closed roles 
rather quickly leads to undecidability, but for sufficiently restricted DLs
we regain decidability.

\begin{theorem} \label{thm:nullability} The following hold:
  \begin{enumerate}
  \item 
    $\qempty(\parTALCHOI, \parCIQ, \parQAQ)$ is in \\
    $\conexptime^\np$;
  \item \label{thm:nullability:undec}
    $\qempty(\parTELIbot, \parCAQ, \parQAQ)$ is undecidable;
  \item $\qempty(\parTLITEOBOOL,\parCAQ,\parQAQ)$ is in \\ $\conexptime^\np$.
  \end{enumerate}
\end{theorem}

\begin{proof}
Given an $\ALCHOI$ ontology $\Oo$, 
$\Sigma\subseteq\conceptnames(\Oo) \cup \rolenames(\Oo)$,
$\Qcwa\subseteq\IQ$, and $q\in\AQ$, we need to check that for each
$\mcI$ over $\Sigma$ with non-empty $\rest(\Oo, \mcI, \Qcwa)$
there exists $\mcJ \in \rest(\Oo, \mcI, \Qcwa)$ with $\answer{\mcJ}{q} = 
\emptyset$. We first show that if this is not the case, then there
is a witnessing instance $\Ii'$ over $\Sigma$ with 
$|\adom{\Ii'}| \leq 2^{|\conceptnames^+(\Oo)|}$.


In what follows, a type $T$ is any subset of $\conceptnames^+(\Oo)$
such that $\top\in T$ and $\bot\not \in T$. For an instance $\Jj$ and
$c\in\const$, the \emph{type of $c$ in $\Jj$} is the set of concepts
$B \in \conceptnames^+(\Oo)$ such that $c\in B^\Jj$.

Let $\mcI$ be an instance over $\Sigma$ such that $\rest(\Oo, \mcI,
\Qcwa) \neq \emptyset$ and $\answer{\mcJ}{q} \neq \emptyset$ for all
$\mcJ \in \rest(\Oo, \Qcwa)$. Let $\mcJ \in \rest(\Oo, \mcI, \Qcwa)$.
Let $\mcI'$ be obtained from $\mcI$ by identifying constants that have
the same type in $\Jj$; when some constants are identified, all
edges incident with them become incident with the resulting
constant; nominals remain themselves.

It is easy to see that $\rest(\Oo, \mcI', \Qcwa)\neq \emptyset$. A
suitable witness $\Jj'$ is obtained from $\mcJ$ by identifying
constants in $\adom{\Ii}$ that have the same type in $\Jj$. By 
construction, $\Jj'$ is an extension of $\mcI'$ that is a model of
$\Oo$ and preserves closed concepts. 

Suppose that $\answer{\Mm'}{q} = \emptyset$ for some $\Mm'\in
\rest(\Oo, \mcI', \Qcwa)$. We shall turn $\Mm'$ into $\Mm \in
\rest(\Oo,\mcI, \Qcwa)$ such that $\answer{\Mm}{q}=\emptyset$.  Let us
replace the copy of $\mcI'$ contained in $\Mm'$ by a copy of
$\mcI$. That is, $\adom{\Mm} =\left ( \adom{\Mm'} \setminus \adom{\Ii'}
\right)\uplus \adom{\Ii}$. We now lift the interpretation of concepts
and roles from $\Mm'$. For $c\in\adom{\mcI}$, let $c'$ be the unique
constant in $\adom{\mcI'}$ that has the same type in $\Jj$ as
$c$. For $c\in\adom{\Mm'} \setminus \adom{\Ii'}$, let $c'=c$. We let
$c\in A^\Mm$ iff $c'\in A^{\Mm'}$ for $A\in\conceptnames(\Oo)$ and  
$(c,d)\in r^\Mm$ iff $(c',d')\in A^{\Mm'}$ for $r \in\rolenames(\Oo)$.
By construction, $\Mm$ is an extension of $\Ii$ that is a
model of $\Oo$, preserves closed concepts, and realizes the same
types. It follows that $\Mm\in \rest(\Oo,\mcI, \Qcwa)$ and
$\answer{\Mm}{q} = \emptyset$.

We can now describe the algorithm. We guess universally an instance
$\Ii$ over $\Sigma$ with $|\adom{\Ii}| \leq
2^{|\conceptnames^+(\Oo)|}$.  We should accept if either $\rest(\Oo,
\Ii, \Qcwa) = \emptyset$ or $\rest(\Oo\cup\{\alpha_q\}, \Ii,\Qcwa)
\neq \emptyset$, where $\alpha_{A(x)}$ is $A\sqsubseteq \bot$ and
$\alpha_{r(x,y)}$ is $\exists r. \top \sqsubseteq \bot$.  Both these
checks amount to deciding if there is a model $\Jj$ of a given ontology
$\Oo'$ that extends a given instance $\Ii$ while preserving a given
set of closed concepts. For $\ALCHOI$ ontologies this problem is in $\np$ in
terms of data complexity~\cite{DBLP:conf/ijcai/LutzSW15}. More
precisely, it can be solved by a nondeterministic algorithm with
running time $\textsf{poly}(|\Ii|)\cdot 2^{\textsf{poly}(|\Oo'|)}$,
which is sufficient to guarantee the $\conexptime^\np$ upper bound. One
way to do this is as follows. Guess the restriction $\Jj_0$ of $\Jj$
to $\adom{\Ii}$. Use type elimination to compute realizable
types. Start from the set of types that are realized in $\Jj_0$ or do
not contain closed concepts. Iteratively remove types that cannot have
their existential restrictions satisfied using available types; for
types realized in $\Jj_0$ ignore existential restrictions that are
already fulfilled in $\Jj_0$. When the set of available types
stabilizes, accept iff it contains each type realized in $\Jj_0$.

Undecidability is obtained via a reduction from the halting problem
for deterministic Turing machines, relying on the ability enforce that
the counter witness to nullability is a grid, and the upper bound for
the third problem is obtained by a a polynomial reduction to the 
first problem (see Appendix for details). 
\end{proof}

Let us now see that the two introduced problems indeed help solve the focusing
problem, and that lower bounds and undecidability propagate to focusing as well. 

\begin{theorem} \label{thm:recognition}
The problems 
\begin{gather*}
  \recognition(\parTALCHOI,\parCIQ,\parFAQ,\parDCQ)\,, \\
  \recognition(\parTALCHOI,\parCAQ,\parFempty,\parDCQ)\,, \\
  \recognition(\parTALCO,\parCempty,\parFempty,\parDCQ)
\end{gather*}
are $\twoexptime$-complete, and
\[\recognition(\parTELIbot,\parCAQ,\parFAQ, \parDempty)\] is
undecidable. 
\end{theorem}

\begin{proof}
We begin with the upper bound for the first problem. 
Let $\Oo$ be an $\ALCHOI$ ontology  and let $\mcF=
(\Sigma,\Qcwa,\Qfix,\Qdet)$ with $\Sigma \subseteq\conceptnames(\Oo)
\cup \rolenames(\Oo)$, $\Qcwa\subseteq\IQ$, $\Qfix \subseteq \AQ$, and
$\Qdet\subseteq \CQ$.
As a first step,  we reduce to the case with only one fixed
query $q_\fix$, that additionally satisfies $\answer{\Oo,
  \emptyset}{q_\fix}=\emptyset$. 
Let $\Oo'$ be the following extension of $\Oo$. For each
$q\in\Qfix$, introduce a fresh concept name $A_q$. If $q=A(x)$ and
$\answer{\Oo, \emptyset}{q} = \left\{(a_1),  \dots,
  (a_k)\right\}$, axiomatize $A_q$ as
\[A_q\equiv A \sqcap \lnot \{a_1, \dots, a_k\}\,.\] 
If $q=r(x,y)$ and 
$\answer{\Oo, \emptyset}{q} = \bigcup _{i=1}^k \left\{(a_i, b_{i,1}),
  \dots, (a_i, b_{i, j_i}) \right\}$, axiomatize $A_q$ as
\[A_q \equiv \left (\lnot \{a_1, \dots, a_k\} \sqcap \exists
  r. \top \right ) \sqcup\bigsqcup_i \{a_i\} \sqcap \exists
  r. \lnot \{b_{i,1},  \dots, b_{i,j_k}\}\,.\] 
Finally, add yet another fresh concept name $B$, axiomatized as
\[ B \equiv \bigsqcup_q A_q\,.\]
Let $q_\fix=B(x)$. By construction, $\mcF$ is a focusing solution
for $\Oo$ iff $\mcF'=(\Sigma,\Qcwa, \{q_\fix\},\Qdet)$ is a focusing
solution for $\Oo'$. Notice that $\Oo'$ 
has polynomial size, but computing it involves finding
$\answer{\Oo, \emptyset}{q}$ for each $q\in\Qfix$.
As we have already mentioned, for $\ALCHOI$ this
is in $\twoexptime$ \cite{CalvaneseEO09}, so we are within 
the intended complexity bounds.

Thus, without loss of generality we can assume that our input is $\Oo$
and $\mcF= (\Sigma,\Qcwa,\{B(x)\},\Qdet)$ such that
$B\in\conceptnames(\Oo)$ and
$\answer{\Oo,\emptyset}{B(x)}=\emptyset$. For such inputs, the first
condition of Definition~\ref{def:focusing} is an instance of
$\qempty(\parTALCHOI, \parCAQ, \parQAQ)$, which we shall prove to be
in $\conexptime^\np \subseteq \twoexptime$
(Theorem~\ref{thm:nullability}).  For the second condition of
Definition~\ref{def:focusing}, we first observe that we can eliminate
the single fixed query $B(x)$ without affecting the set of intended
models: it suffices to add the concept inclusion $B \sqsubseteq\bot$
to $\Oo$. Then, to reduce the second condition to mixed query
answering, we first introduce fresh duplicates $A'$ and $r'$ for all
$A\in\conceptnames(\Oo)$ and $r\in\rolenames(\Oo)$ that do not occur
in $\Qcwa$. Next, we construct $\Oo'$ and $q'$ for all $q\in\Qdet$ by
replacing original predicates in $\Oo$ and $q$ with their
duplicates. It is straightforward to see that $\answer{\mcJ_1}{q} =
\answer{\mcJ_2}{q}$ for all $\mcI$ and $\mcJ_1,\mcJ_2 \in
\semantics(\Oo, \mcF, \mcI)$ iff $q\subseteq q'$ over all models of
$\Oo\cup\Oo'$ with finite predicates $\Qcwa$. Query containment can be
reduced to query answering in the usual way, by incorporating the
query $q$ into the ontology using nominals; the presence of finiteness assumptions (in
both problems) does not affect the construction. This way we have
reduced the second condition to $|\Qdet|$ polynomial-size instances of
mixed query answering. As we have argued, the latter is in
$\twoexptime$ for $\ALCHOI$, which gives the desired upper bound for
the focusing problem.

To obtain the upper bound for the second problem, note that 
in the absence of fixed predicates the first condition of
Definition~\ref{def:focusing} trivializes and the second condition can
be checked exactly like before, but without the first preprocessing step.


For the third problem the upper bound follows directly from either of
the previous cases. To show the lower bound, we reduce from the
standard unrestricted query answering problem for $\ALCO$
\cite{NgoOS16}. Let $\Oo$ be an $\ALCO$ ontology and let $q\in \CQ$. Let
$\Oo'$ be the ontology obtained from $\Oo$ by introducing a fresh concept
name $A$, axiomatized with
$\{c\} \sqsubseteq A$ and $A \sqsubseteq \forall
r. A$ for all $\{c\}\in\conceptnames^+(\Oo)$,
$r\in\rolenames^+(\Oo)$, and replacing each concept inclusion $C
\sqsubseteq D$ with $C \sqcap A \sqsubseteq D$. Each model $\mcI$
of $\Oo$ can be turned into a model of $\Oo'$ by letting $A^\Ii =
\adom{\Ii}$;  each model $\Ii'$ of $\Oo'$ can be turned into a
model of $\Oo$ by restricting it to $A^{\Ii'}$. Consequently,
$\answer{\Oo,\emptyset}{q} = \answer{\Oo',\emptyset}{q}$, but there
is a model $\Ii'$ of $\Oo'$ that satisfies $q$, because outside of
$A$ no restrictions are imposed by $\Oo'$. It follows that
$(\emptyset, \emptyset, \emptyset, \{q\})$ is a focusing solution for
$\Oo'$ iff $\Oo$ entails $q$.

For the last claim, note that if $\answer{\Oo, \emptyset}{q} =
\emptyset$, then the following are equivalent:
\begin{itemize}
\item $(\Oo, \Sigma,  \Qcwa, q)$ is a positive instance of $\qempty$;
\item $(\Oo, (\Sigma,  \Qcwa,  \{q\}, \emptyset))$ is a positive
  instance of $\recognition$.
\end{itemize}
Because the reduction in the proof of
Theorem~\ref{thm:nullability}~\eqref{thm:nullability:undec} gives a
query that has the above property, we can conclude that this variant
of the focusing problem is undecidable.  
\end{proof}


%% file: 53-entailment.tex
\subsection{Entailment}\label{sec:entailment}

The goal of this section is to prove the following theorem.

\begin{theorem} The problem 
\[\entailment(\parTALCHOI,\parCCQ, \parFempty, \parQCQ)\] is
\twoexptime-complete in combined complexity, and \conp-complete in
data complexity.
\end{theorem}

Let $\dlkb$ be an $\mathcal{ALCHOI}$ KB, $\Qcwa$ a
set of CQs, $q$ a Boolean CQ, and $\Ii$ an instance.  Unless specified
otherwise, we will always consider $\Sigma_\dlkb$-instances where
$\Sigma_\dlkb$ is the signature of $\dlkb$, i.e., the set of concept
and role names occurring in $\dlkb$.  Our goal is to decide if for
every instance $\Jj \in \rest(\dlkb,\Ii,\Qcwa)$, $\Jj \models q$.
This is clearly equivalent to the problem of finding a counter model
for $q$: an instance $\Jj \in \rest(\dlkb,\Ii,\Qcwa)$ such that
$\Jj \nvDash q$.  Moreover, it suffices to consider counter models
that are almost forests. An instance $\Jj$ is a \emph{tree
  extension} of $\Ii$ if $\Jj = \Jj_0 \cup \Jj_1$ such that
$\adom{\Jj_0} = \adom{\Ii}$ and $\Jj_1$ is a collection of trees of
bounded degree in which elements of $\adom{\Ii}$ occur only in
leaves (and never in roots, even if they are also leaves). Note that
the partition of $\Jj$ into $\Jj_0$ and $\Jj_1$ is unique. We 
call $\Jj_1$ the \emph{forest} of $\Jj$.

\begin{lemma}\label{lem:unravelling} The following are equivalent:
 \begin{enumerate} 
  \item There exists  $\Jj \in \rest(\dlkb,\Ii,\Qcwa)$ such that
    $\Jj \nvDash q$;
  \item There exists $\Tt  \in \rest(\dlkb,\Ii,\Qcwa)$  such that $\Tt
    \nvDash q$ and $\Tt$ is a tree extension of $\Ii$.   
  \end{enumerate}
\end{lemma}

Next we observe that $ \answer{\Jj}{q'} =\answer{\Ii}{q'}$ for every
query $q' \in \Qcwa$ can be reformulated as a non-entailment problem
of a UCQ capturing 'bad matches' of queries in $\Qcwa$.  Intuitively,
a match $\pi$ of a query $q'(\vec{x})$ in $\Qcwa$ is \emph{bad} if
$\pi(\vec{x}) \notin \answer{\Ii}{q'}$, or if it maps an answer
variable to some $d$ not in $\adom{\Ii}$.
To this aim, we assume below that we have two concept names $A_0$
and $\bar A_0$ such that for each instance $\Jj$ in
$\rest(\dlkb,\Ii,\Qcwa)$ we have $A_0 ^\Jj= \adom{\Ii}$ and $\bar A_0
^\Jj= \adom{\Jj}\setminus\adom{\Ii}$.
 \begin{lemma} 
For every $q(\vec{x}) = \exists \vec{y}\,\varphi(\vec{x},\vec{y})$ in
$\Qcwa$,
let \[\widehat{q} = \!\!\bigvee_{\vec{a} \notin \answer{\Ii}{q}} \!\!\!
\exists\vec{y}\,\varphi(\vec{a},\vec{y}) \; \lor \bigvee_{x \in
  \vec{x}} \exists \vec{x}\,\exists \vec{y}\,
\bar{A}_0(x) \land \varphi(\vec{x},\vec{y}) \;\; \text{and} \;\;
\widehat{Q} = \!\!\bigvee_{q \in \Qcwa} \!\!\widehat{q}\,.\]   Then $\Jj \in
\rest(\dlkb,\Ii,\Qcwa)$ iff $\Jj \models (\dlkb, \Ii)$ and $\Jj
\nvDash \widehat{Q}$.
\end{lemma}
Putting all pieces together, the problem thus reduces to deciding the
existence of a tree-like model of $\dlkb$ that extends $\Ii$ and is a
counter model for the UCQ $\widehat{Q} \lor q$.

Using an approach similar to that used
e.g. in~\cite{DBLP:journals/jcss/EiterOS12,DBLP:conf/ijcai/LutzSW15}
we will start by establishing the \conp upper bound in data complexity
by decomposing counter models and then use a guess and check algorithm
for finding such decompositions.

 
Given an instance $\Jj$ and an element $d \in \const$, the
\emph{$\dlkb$-type of $d$ in $\Jj$} is defined as $\mathsf{tp}_\Jj(d)
= \left\{ A \in \conceptnames(\dlkb) \mid d \in A^\Jj \right\}$.  A
\emph{(realizable) unary type} $\tau$ for $\dlkb$ is a subset of
$\conceptnames(\dlkb)$ such that there is a model $\Jj$ of $\dlkb$ and
$d \in  \const$ with $\tau = \mathsf{tp}_\Jj(d)$.  Further, for
unary types $\tau, \tau'$ and role $r$ we write $\tau \leadsto_r \tau'
$ if there is a model $\Jj$ of $\dlkb$ and $d, e\in\const$ such that
$(d,e) \in r^\Jj$, $\mathsf{tp}_\Jj(d) = \tau$ and
$\mathsf{tp}_\Jj(e)= \tau'$.  We now extend the notion of types to
capture small substructures of tree extensions of $\Ii$.
 
\begin{definition} Let $n \geq 1$.  A $n$-type $\Mm$ for $(\dlkb,\Ii)$
  is a finite tree  extension of $\Ii$ such that 
\begin{enumerate}
\item the forest of $\Mm$ is a single tree of out-degree at most $|\dlkb|$ and  depth at most~$n$;
\item for each role $r$, and  all $(d,d') \in r^\Mm$, $\mn{tp}_{\Mm}(d) \leadsto_r \mn{tp}_{\Mm}(d')$; 
\item for each $d \in A^\Mm \setminus \adom{\Ii}$ at depth at most
  $n-1$ and each $A \sqsubseteq \exists r. B$ in $\dlkb$ there
  exists $e \in B^{\Mm}$ such that $(d,e) \in r^\Mm$; 
\item $r^\Mm \subseteq s^\Mm$  for all $r \sqsubseteq s$ in $\dlkb$.
\end{enumerate}
We write $\mn{root}(\Mm)$ for the root of the unique tree in the
forest of $\Mm$.
\end{definition}

For an element $d \in \adom{\Mm}$, we use $\Mm(d)_{k}$ to denote the
subinstance of $\Mm$ induced by $\adom{\Ii}$ and the subtree of depth
at most $k$ of the only tree in the forest of $\Mm$, rooted at $d$.  We write $\Mm(d)_{ k} \simeq
\Mm'(e)_ {k}$ if there is an isomorphism $g$ from $\Mm(d)_{k}$ to
$\Mm'(e)_ {k}$ such that $g(d)= e$.

\begin{definition} A set $\Gamma$ of $n$-types for $(\dlkb,\Ii)$ is
  \emph{coherent} if the following  are satisfied: 
\begin{enumerate}
\item for all $\Mm, \Mm' \in \Gamma$, $\adom{M} \cap \adom{M'} =
\adom{\Ii}$ and $\Mm |_{\adom{\Ii}} = \Mm' |_{\adom{\Ii}}$;
\item for every $c \in A^\Mm\cap \adom{\Ii}$ and $A \sqsubseteq
\exists r. B$ in $\dlkb$ there exists $\Mm \in \Gamma$ such that there
is an $r$-edge from $c$ to $\mn{root}(\Mm)$ and $\mn{root}(\Mm) \in B^{\Mm}$;
\item for every $\Mm \in \Gamma$ and every successor $d$ of
$\mn{root}(\Mm)$, there is $\Mm' \in \Gamma$ such that $\Mm(d)_{ n-1}
\simeq \Mm'(e)_ {n-1}$ with $e = \mn{root}(\Mm')$.
\end{enumerate}
\end{definition}

\begin{lemma}[\cite{DBLP:journals/jcss/EiterOS12,DBLP:conf/ijcai/LutzSW15}]
 Let $Q$ be union of Boolean CQs, each using at most $n$
 variables. Then 
 $(\dlkb,\Ii) \models Q$ iff  for each coherent set $\Gamma$ of
 $n$-types for $(\dlkb,\Ii)$ it holds that  $\bigcup_{\Mm \in \Gamma} \Mm \models Q$.
\end{lemma}

Thus, we can test whether $(\dlkb, \Ii) \models Q$, by universally
guessing a set of coherent types $\Gamma$ and verifying that
$\bigcup_{\Mm \in \Gamma} \Mm \nvDash Q$.
The size of an $n$-type for $(\dlkb, \Ii)$ is at most $|\dlkb|^n +
|\Ii|$. Because all $n$-types in a coherent set coincide over
$\adom{\Ii}$, up to isomorphism there is at most $ (2+|\Ii|)^{|\dlkb|^{\mn{poly}(n)}}$
different $n$-types. Hence, we can impose the same bound on the size the
sets $\Gamma$ guessed in the algorithm.  
Further, checking whether a coherent set of $n$-types satisfies a
union $Q$ of Boolean CQs, each using at most $n$ variables can be done
in time $\mn{poly}(|Q|, |\Ii|^n)$.
We apply this algorithm to $Q=\widehat Q$. We can take for $n$ 
the maximal number of variables in queries from $\Qcwa$. It then
follows that the size of each $\widehat{q}$  is $O(|q|\cdot
|\adom{\Ii}|^n + |q|\cdot n)$, and $\widehat Q$ is the 
union of $|\Qcwa|$ such queries. This gives the $\conp$ upper bound
for the data complexity of the entailment problem. The lower bound follows 
from IQ entailment in a sublogic of $\mathcal{ALCHOI}$~\cite{DBLP:journals/jiis/Schaerf93}.

To establish the \twoexptime-completeness in combined complexity,
observe that the algorithm in~\cite{CalvaneseEO09} for deciding
entailment of a UCQ $Q$ in $\ALCHOI$ runs in time $2^{(|\Oo| + |\Ii| +
|Q|)^{\mn{poly}(n)}}$, where $n$ is the maximal number of variables in any
CQ constituting $Q$. The lower bound follows from CQ entailment in
$\ALCI$~\cite{DBLP:conf/cade/Lutz08}.


%% file: 54-emptiness.tex
\subsection{Emptiness and Mixed
  Satisfiability}\label{sec:emptiness}

In this section we present a collection of complexity results for the
$\emptiness$ problem.  As a technical tool, we introduce another
closely related problem, called \emph{mixed satisfiability}, or
$\mixedsat$. Similarly to $\mixedqa$, $\mixedsat$ corresponds to a stronger
version of unrestricted satisfiability of a theory, but a relaxed version
of finite satisfiability: we are looking for models of the input
theory in which all predicates from a given set have finite extensions (the
remaining predicates may have a finite or an infinite extension). This
problem will not only make the characterization of $\emptiness$
clearer, but it is interesting in its own right. For instance, given a
theory $\varphi$ and a set $\Sigma$ of predicate symbols, we can use
$\mixedsat$ to check whether there exists a finite database $D$ over
$\Sigma$ such that $D$ can be extended to a model of $\varphi$ while
preserving $D$, i.e., by finding a finite or an infinite extension for
all predicates of $\varphi$ that do not appear in $\Sigma$. Formally,
$\mixedsat$ is defined as follows.

\problemdef{$\mixedsat(\Lth)$}{A pair $(\varphi,\Sigma)$ with
  $\varphi\,{\in}\, \Lth$ and $ \Sigma\subseteq \rels$.}{Does exist
  a model $\mcJ$ of $\varphi$ such that $R^{\mcJ}$ is finite for all
  $R\in \Sigma$?}

We let $\mixedunsat(\Lth)$ be the complement of the decision problem
$\mixedsat(\Lth)$.

\smallskip
We next observe that if we consider only atomic closed queries (closed
predicates) and we do not require predicate fixing, then $\mixedunsat$ and
$\emptiness$ collapse into one problem.
 
\begin{proposition}\label{prop:mixed-reduction} For each $\Lth$ there exist polynomial
  reductions
  \begin{enumerate}    
  \item from $\emptiness(\parTL,\parCAQ, \parDany,\parFempty)$ \\ to
    $\mixedunsat(\Lth)$, and \smallskip
  \item from $\mixedunsat(\Lth)$\\ to
    $\emptiness(\parTL,\parCAQ, \parDempty,\parFempty)$.
  \end{enumerate}
\end{proposition}

Recall that a fragment $\Lth$ of first-order logic has the
\emph{finite model property (FMP)}, if any satisfiable theory
$\varphi\in \Lth$ has a finite model. Clearly, for any
$\Lth$ with the FMP, $\mixedsat(\Lth)$ corresponds to ordinary
satisfiability in $\Lth$, i.e. $(\varphi,\Sigma)$ is a positive
instance of $\mixedsat(\Lth)$ iff $\varphi$ is satisfiable. Here, we
look at some fragments that do not have the FMP: we concentrate on
mixed satisfiability in DLs that support inverse
roles, and simultaneously allow to express functionality of
roles. These features cause the loss of the FMP, and make mixed
satisfiability non-trivial.

\begin{example} 
  Consider the following  ontology~$\,\Oo$: 
  \begin{align*}
   A&\ISA \exists r. B\,, &   B&\ISA \exists r. A\,,   &  \exists p^{-}\ISA \exists p\,,\\ 
   A&\ISA \exists p \sqcap \neg \exists .{p^{-}}\,, &  \{d\}&\ISA A\,,     &   \func(p^{-})\,. 
  \end{align*}
  Let $c_1,c_2,c_3,\ldots$ be an enumeration of $\const$ with \mbox{$c_1=d$.} It is easy to see that $(\Oo,\{A,B,r\})$ is a positive instance of 
  $\mixedsat$ for $\ALCHOIF$, which is witnessed by a model
  $\mcI$ with $A^{\mcI}=\{c_1\}$, $B^{\mcI}=\{c_2\}$,
  $r^{\mcI}=\{(c_1,c_2)\}$, and
  $p^{\mcI}=\{(c_{i},c_{i+1})\mid i\in\nat\}$. Note that in such
  $\mcI$, only $p$ has an infinite extension. Observe  that
  $\Oo$  forces $p$ to be infinite, i.e. $(\Oo,\Sigma)$ is a
  negative instance of $\mixedsat$ for $\ALCHOIF$ for any $\Sigma$ that
  contains $p$.
 \end{example}

 We next study mixed satisfiability for the DL $\ALCHOIF$ and its
 sublogic $\ALCHIF$, and show that the problem for these DLs is
 complete for \nexptime and \exptime, respectively. We then argue that
 for some important DLs of the $\DLLITE$ family the complexity can be
 lowered significantly.

 \smallskip

We first observe that, in all DLs considered here, role names can be
eliminated in polynomial time from $\Sigma$ in instances
$(\Oo,\Sigma)$ while preserving mixed satisfiability.  To see this,
suppose $(\Oo,\Sigma)$ is an instance of mixed satisfiability with
$\Oo$ in one of the above DLs, and suppose $\Sigma'$ is the set of
role names in $\Sigma$. We let
 \[\Oo'=\Oo\cup\left\{\top \ISA \forall r.A \,, \;\top \ISA \forall
r^{-}.A \mid r \in\Sigma'\right\}\,,\] where $A$ is a fresh concept
name that does not occur in $\Oo$. Intuitively, we use $A$ to collect
the domain and the co-domain of every role name in $\Sigma$. It is
easy to see that $(\Oo,\Sigma)$ is a positive instance of mixed
satisfiability iff $(\Oo,\{A\}\cup \Sigma\setminus \Sigma')$ is a
positive instance of mixed satisfiability.

 Keeping in mind the above observation, our goal next is to provide an
 algorithm to decide, given a pair $(\Oo,\Sigma)$, where $\Oo$ is an
 ontology expressed in $\ALCHOIF$ and $\Sigma$ is a set of
 \emph{concept names}, whether there exists a model $\mcI$ of $\Oo$
 such that $A^{\mcI}$ is finite for all $A\in \Sigma$.  Our algorithm
 for $\ALCHOIF$, as well as the algorithms for the remaining DLs,
 build on the previously developed methods for finite model
 reasoning in DLs. Specifically, there exist algorithms that can solve
 $\mixedsat(\ALCHOIF)$ and $\mixedsat(\ALCHOIF)$ for instances
 $(\Oo,\Sigma)$, where $\Sigma$ must be the set of \emph{all} concept
 and roles names in $\Oo$. For us the most relevant
 are~\cite{DBLP:journals/iandc/LutzST05,DBLP:journals/logcom/Pratt-Hartmann07,DBLP:journals/jolli/Pratt-Hartmann05},
 which present algorithms based on reductions to \emph{integer
   programming}. Our algorithm and the presentation can be described
 as follows. We define the notions of a \emph{tile} $T$ for $\Oo$ and
 a \emph{mosaic} $N$ for $(\Oo,\Sigma)$. Note that DLs considered here
 support only unary and binary relations, and thus models here can be
 naturally seen as directed graphs with labels on nodes (capturing the
 content of concept names), and labels on edges (describing the
 extensions of role names).  Intuitively, a tile $T$ is a (finite and
 small) description of a single domain element together with its
 (relevant) neighborhood in a model of $\Oo$. A tile can be seen as a
 building block for constructing models of $\Oo$. A mosaic $N$ is a
 multiset of tiles that has to satisfy certain coherence
 conditions. The conditions ensure that a desired model of $\Oo$ can
 be assembled by instantiating tiles in $N$ (according to the tile
 multiplicities given by $N$). In the final step, we show that
 deciding the existence of a mosaic for $(\Oo,\Sigma)$ can be reduced
 to a variant of integer programming problem, which, as we argue, can
 be solved algorithmically and yields worst-case optimal upper bounds
 for $\ALCHOIF$ and $\ALCHIF$. We concentrate here on ontologies in normal form.


  \smallskip
 We start with the formal definition of tiles.

 \begin{definition}[Tiles]
   Let $\Oo$ be  an $\ALCHOIF$ ontology in normal form. A \emph{type $T$ (for $\Oo$)}
   is any subset of $\simpleconcepts(\Oo)$
   such that $\top\in T$ and
   $\bot\not \in T$. We use $\mathsf{Types}(\Oo)$ to denote the set of
   types for $\Oo$. A \emph{tile $\tau$ (for $\Oo$)} is any
   tuple $\tau= (T,\rho)$, where $T\in \mathsf{Types}(\Oo)$,  $\rho$
   is a set of pairs $(R,T')$ with $R\subseteq \roles(\Oo)$ and $T'\in
   \mathsf{Types}(\Oo)$, and such 
   that the following conditions~are~satisfied:
   \begin{enumerate}
   \item $|\rho|\leq |\Oo|$;


   \smallskip\item If $B_1\sqcap \ldots \sqcap B_{k-1}
   \ISA B_{k}\sqcup\ldots \sqcup B_m \in \Oo$ and \linebreak
      \mbox{$ \{B_1,\ldots, B_{k-1}\}\subseteq T$}, then $
      \{B_{k}, \ldots,  B_m \}\cap T\neq\emptyset$;

   \smallskip\item If $A\ISA \exists r. B\in \Oo$ and $A\in T$, then there is
     $(R,T')\in \rho$ such that $r\in R$ and $B\in T'$;
   \smallskip\item For all $(R,T')\in \rho$, the following hold:
     \begin{enumerate}
     \smallskip\item If $A\ISA \forall r. B\,{\in}\,\Oo$, $ A\in T$ and $r\in R$,
       then $B\in T'$;
     \smallskip\item If $A\ISA \forall r. B\,{\in}\,\Oo$, $ A\in T'$ and
       $r^{-}\,{\in}\,R$, then $B\,{\in}\,T$;
       \smallskip\item If $r\ISA s \in \Oo$ and $ r\in R$, then $s\in R$;
       \smallskip\item If
       $r\ISA s \in \Oo$ and $ r^{-}\in R$, then $s^{-}\in R$.
     \end{enumerate}
    
   \smallskip\item If $\func(r)\in\Oo $, then
     $|\{(R,T)\in \rho\mid r\in R \}|\leq 1$.

   \end{enumerate}
   \smallskip
   We use $\tiles(\Oo)$ to denote the set of all tiles for $\Oo$.
 \end{definition}
 
 Intuitively, a tile $\tau= (T,\rho)$ for $\Oo$ describes an element
 $e$ that participates in basic concepts given by $T$, and whose
 neighborhood is described (at least partially) by $\rho$. Each
 $(R,T')\in \rho$ can be seen as a labeled edge outgoing from $e$. The
 sets $T$ and $\rho$ form a configuration that is consistent with the
 statements in $\Oo$ (from the `perspective' of $e$):
 $e$~participates in suitable simple concepts (Condition~(2)), there
 is a proper witness edge for all existential inclusions
 $A\ISA \exists r.B$ that are `triggered' by $e$ (Condition~(3)),
 all edges outgoing from $e$ satisfy all universal inclusions
 $A\ISA \forall r.B$ and all role inclusions (Condition~(4)), and the
 edges described by $\rho$ are compatible with the functionality
 assertions in $\Oo$ (Condition~(5)).

 \smallskip

 We now formally define mosaics, which are finite representations of
 desired models. Intuitively, a mosaic tells us how many instances of
 different tiles we need in order to construct a desired model. Since
 in our setting some tiles might have to be instantiated infinitely
 many times, we need a value to talk about the cardinality of a
 countable infinite set. In dealing with this, we
 closely follow~\cite{DBLP:journals/jolli/Pratt-Hartmann05}, which shows
 how an algorithm for (fully) finite satisfiability in $\mathcal{C}^{2}$
 can be turned into an algorithm for (fully) unrestricted satisfiability. We let
 $\natstar=\nat\cup\{\alnull\}$. Since mosaics will involve linear
 inequalities, the usual ordering $<$, and the arithmetic operations
 $+$ and $\cdot$ on $\nat$ need to be extended to accommodate the new
 element $\alnull$. In particular, $n < \alnull$ for all $n\in
 \nat$. We let $\alnull \cdot \alnull = \alnull + \alnull =
 \alnull$. We let $\alnull + n = n + \alnull = \alnull $ for all
 $n\in \nat$. We let $\alnull \cdot n = n \cdot \alnull = \alnull $
 for all $n\in \nat$ with $n>0$. Finally, we let
 $\alnull \cdot 0 = 0\cdot \alnull = 0$.
 
 \begin{definition}[Mosaics]\label{def:mosaic} Assume an ontology
   $\Oo$, and a set $\Sigma$ of concept names. Given a set $R$ of
   roles, we let $R^{-}= \{r^{-}\mid r\in R\}$. A \emph{mosaic} for
   $(\Oo,\Sigma)$ is a function $N:\tiles(\Oo)\rightarrow \natstar$
   that satisfies the following conditions:
   \begin{enumerate}

   \item \label{item:nominals} For every nominal
     $\{c\}\in \simpleconcepts(\Oo)$ we have:
     \[\sum_{\substack{(T,\rho)\in\tiles(\Oo) \,\land \\ \{c\}\in T}}
       N((T,\rho)) =1 \,;\]

   \item \label{item:sometile} The following inequation is
     satisfied: \[\sum_{\tau\in\tiles(\Oo)} N(\tau) \geq 1 \,;\]

   \item \label{item:finiteness} For every $A\in \Sigma$ and every
     tile $(T,\rho)\in \tiles(\Oo)$ with $A\in T $, we have
     $N((T,\rho))\neq \alnull$;

     \medskip
   \item \label{item:successor} For all $(T,\rho)\in\tiles(\Oo) $
     and $(R,T')\in \rho $ the following implication holds: if $N((T,\rho))> 0$, then there is some $\rho'$ such that  $(T',\rho')\in \tiles(\Oo)$ and $N((T',\rho'))>0$;

     \medskip
      \item\label{item:one-way-matching} For every pair
     $T,T'\in \mathsf{Types}(\Oo)$ and every
     $R\subseteq \roles(\Oo)$ with  $\func(r^{-})\in\Oo$, the following inequation is satisfied:
     \[\sum_{ \substack{(T,\rho)\in \tiles(\Oo)\,\land \\ (R,T')\in
           \rho } }N((T,\rho)) \leq
       \sum_{\substack{\substack{(T',\rho')\in \tiles(\Oo)\,\land
             \\(R^-,T)\in \rho' }}}N((T',\rho')) \;.
     \]

   \end{enumerate}
 \end{definition}
 
 The conditions in Definition~\ref{def:mosaic} are geared to ensure
 that a mosaic $N$ for $(\Oo,\Sigma)$ witnesses the existence of a
 model $\mcI$ for $\Oo$ where each concept name from $\Sigma$ has a
 finite extension. Condition (\ref{item:nominals}) tells that in a
 candidate for $\mcI$ there is exactly one element that satisfies a
 nominal $\{c\}$ (which must be the constant $c$ itself). Condition
 (\ref{item:sometile}) tells that at least one tile needs to be
 realized in a model of $\Oo$. Condition (\ref{item:finiteness})
 ensures that tiles $(T,\rho)$ with $\Sigma\cap T\neq \emptyset $ are
 allowed to be instantiated only finitely many times. Condition
 (\ref{item:successor}) ensures that for tile $(T,\rho)$ that will be
 instantiated at some element $e$, we can also instantiate tiles to
 provide neighbors for $e$ as prescribed by $\rho$. Finally, to
 understand the Condition~\ref{item:one-way-matching}, let us take a pair
 $T,T'\in \mathsf{Types}(\Oo)$ and a set $R\subseteq \roles(\Oo)$, and
 let us view a candidate for $\mcI$ as labeled graph. Suppose the
 graph has $n$ nodes that are labeled with $T$, and having an
 $R$-labeled edge to a node labeled with $T'$. If
 $\func(r^{-})\in\Oo$, then clearly there must exists at least $n$ nodes
 that are labeled with $T'$, and having an $R^-$-labeled edge to a
 node labeled with $T$.

 The conditions  placed on mosaics correctly characterize mixed
 satisfiability in $\ALCHOIF$ (see Appendix for a proof).
 
\begin{theorem}\label{thm:mosaic}
 Assume an $\ALCHOIF$ ontology $\Oo$, and a set $\Sigma\subseteq
 \conceptnames$. Then the following statements  are equivalent: 
  \begin{enumerate} 
  \item $\Oo$ has a model $\mcI$ such $A^{\mcI}$ is finite for all
     $A\in \Sigma$.
  \item There exists a mosaic $N$ for $(\Oo,\Sigma)$.
  \end{enumerate}
\end{theorem}

Due to Theorem~\ref{thm:mosaic}, mixed satisfiability in $\ALCHOIF$ reduces
to deciding the existence of a mosaic. To reason about mosaics, we
 employ (an extension of) integer programming. In particular, we
consider linear inequations of the~form
\begin{equation}
  a_1\cdot x_1+\cdots+a_n\cdot x_n + c \leq b_1\cdot
  y_1+\cdots+b_m\cdot y_m\label{eq:linear}
\end{equation}
where $a_1,\ldots,a_n,b_1,\ldots,b_m$ are positive integers, $c$ is a
(possibly negative) integer, and $x_1,\ldots,x_n,y_1,\ldots,y_m$ are
variables. In case $c\geq 0$, the inequation \eqref{eq:linear} is a
called \emph{positive}.  An \emph{ordinary inequation system} is a
pair $(V,\mathcal{E})$, where $\mathcal{E}$ is a set of linear
inequations of the form \eqref{eq:linear}, and $V$ is the set of
variables that appear in $\mathcal{E}$. The notion of a
\emph{solution} $S: V\rightarrow \nat$ for $(V,\mathcal{E})$ is the
standard one. The definition of solutions \emph{over} $\nat$ can be
lifted to solutions \emph{over} $\natstar$ in the obvious way, i.e.,
we will consider solutions to $(V,\mathcal{E})$ that are functions of
the form $S: V\rightarrow \nat^{*}$. An \emph{enriched inequation
  system} is a tuple $(V,\mathcal{E},F,I)$, were $(V,\mathcal{E})$ is
an ordinary inequation system, $F\subseteq V$, and $I$ is a set of
\emph{implications}, which are expressions of the form
\[y_1+\cdots +y_m>0 \Rightarrow x_1+\cdots + x_n>0\] with
$y_1,\ldots,y_k,x_1,\ldots,x_n\in V$. A solution $S$ (over $\nat$ or
over $\natstar$) to $(V,\mathcal{E},F,I)$ is a solution to
$(V,\mathcal{E})$ that additionally satisfies the following
conditions:
\begin{enumerate}
\item $S(x)\neq \alnull$ for all $x\in F$;
\item $S(x)>0$ implies $S(x_1)+\cdots + S(x_n)>0$ for all \\
  $x>0 \Rightarrow x_1+\cdots + x_n>0$ in $I$.
\end{enumerate}

The following result on the complexity of reasoning with enriched
inequation systems can be shown (see Appendix).

  \begin{theorem}\label{thm:enriched-ineq}
    Deciding the existence of a solution for an enriched inequation
    system $\mathcal{H}=(V,\mathcal{E},F,I)$ is feasible in
    non-deterministic polynomial time in the size of $\mathcal{H}$. If
    $\mathcal{E}$ contains only positive inequations, the problem is
    solvable in polynomial time in the size of $\mathcal{H}$.
  \end{theorem}

  Assume an $\ALCHOIF$ ontology $\Oo$, and a set
  $\Sigma\subseteq \conceptnames$. Due to Theorem~\ref{thm:mosaic},
  checking the existence of a model $\mcI$ of $\Oo$ where $A^{\Sigma}$
  is finite for all $A\in \Sigma$ reduces to checking the existence of
  a mosaic for $(\Oo,\Sigma)$. It is easy to see that the latter can
  be reduced to checking the existence of a solution to an
  exponentially sized enriched inequation system
  $(V^*,\mathcal{E}^*,F^*,I^*)$. Intuitively, $V^*$ is obtained by
  taking a variable $x_\tau$ for every $\tau\in \tiles(\Oo)$. The
  inequations in $\mathcal{E}^{*}$ are obtained directly from the
  expressions in points (\ref{item:nominals}), (\ref{item:sometile})
  and (\ref{item:one-way-matching}) of Definition~\ref{def:mosaic}, by
  replacing $N(\tau)$ with the variable $x_\tau$ for all
  $\tau\in\tiles(\Oo)$.  The set $I^*$ is obtained directly from point
  (\ref{item:successor}). That is, for all $(T,\rho)\in\tiles(\Oo) $ and
  $(R,T')\in \rho $, $I^*$ contains
  \[x_{(T,\rho)}>0 \Rightarrow x_{\tau^1}+\cdots + x_{\tau^m}> 0\,,\]
  where
  $\{\tau^1,\ldots , \tau^m\} =\{(T',\rho')\in\tiles(\Oo)\mid \mbox{
    for some }\rho'\}$. Finally,
  $F^*=\{x_\tau\mid \tau = (T,\rho)\in \tiles(\Oo)\land \Sigma\cap
  T\neq\emptyset \}$. We observe here that in the case $\Oo$ is an 
  $\ALCHIF$, we can ignore point (\ref{item:nominals}) in the
  construction of $ \mathcal{E}^{*}$, which then results in a set that
  contains only positive inequations.

As the above encoding requires exponential time, the complexity
results in Theorem~\ref{thm:enriched-ineq} lead to the following
complexity bounds for $\ALCHOIF$ and $\ALCHIF$.

\begin{theorem}\label{thm:empt-complexity} The following hold:
  \begin{enumerate}
    \itemsep0.2em 
  \item $\mixedsat(\parTALCHOIF)$ is \nexptime-complete. 
  \item $\mixedsat(\parTALCHIF)$ is \exptime-complete. 
  \item $\emptiness(\parTALCHOIF,\parCAQ, \parDany,\parFempty)$ \\ is
    \conexptime-complete. 
  \item
    $\emptiness(\parTALCHIF,\parCAQ, \parDany,\parFempty)$ \\ is
    \exptime-complete.
  \end{enumerate}
\end{theorem}
\noindent
The lower bounds for the above problems are inherited from general
satisfiability in $\ALCHOIF$~\cite{Tobies2000} and
$\ALCHIF$~\cite{Schild1991}.

\medskip We now turn our attention of the $\DLLITE$ family of DLs,
and, in particular, to the DLs $\DLLITEBOOL$ and $\DLLITEFUNC$, which are
important members of this family without the FMP.  
Their syntactic restrictions open the way to different algorithms
and allow us to obtain further complexity results. 

\begin{theorem} The following are true:
  \begin{enumerate}
    \itemsep0.2em
  \item $\mixedsat(\parTLITEBOOL)$ is \np-complete.
  \item $\mixedsat(\parTLITEFUNC)$ is in \ptime. 
  \item $\emptiness(\parTLITEBOOL,\parCAQ, \parDany,\parFempty)$ \\ is
    \conp-complete. 
  \item  $\emptiness(\parTLITEFUNC,\parCAQ, \parDany,\parFempty)$ \\ is
    \ptime-complete.
  \end{enumerate}
\end{theorem}
The upper bounds for $\parTLITEBOOL$ above are obtained by modifying
our reduction to (enriched) integer programming, while the lower
bound is inherited from satisfiability in propositional logic. The
upper bound for $\DLLITEFUNC$ can be obtained by applying the
\emph{cycle reversion} technique~\cite{DBLP:conf/esws/Rosati08}. In
particular, using polynomial time computation mixed satisfiability in
$\DLLITEFUNC$ can be reduced to ordinary satisfiability, which is
known to be tractable.

The above results on the $\emptiness$ problem deal with the case where
$\parFempty$. We can additionally deal with the case $\parFAQ$ in DLs
that allow to freeze the interpretation of concept and role names to
some given extensions, enabling us to reduce the case where
$\Qfix\neq\emptyset$ to the setting where $\Qfix=\emptyset$.  This can
be done in $\ALCHOIF$ as shown in the proof of
Theorem~\ref{thm:recognition}, and the same can be done in
$\DLLITEBOOL$. From this we can infer the following upper bounds.


\begin{theorem} The following hold:
  \begin{enumerate}
  \item $\emptiness(\parTALCHOIF,\parCAQ, \parDany,\parFAQ)$ \\ is in
    $\ptime^{\nexptime}$. \smallskip
  \item $\emptiness(\parTLITEBOOL,\parCAQ, \parDany,\parFAQ)$ \\ is in
    $\ptime^{\np}$.
  \end{enumerate}
  \end{theorem}


%% file: 6-discussion.tex
\section{Related Work}\label{sec:related-work}

\subsubsection*{Module Extraction} Focusing as defined in this paper is
related to \emph{module extraction}, which has been studied
extensively in DLs as a tool to facilitate the development and use of
very large ontologies. Various notions of modularity and algorithms
for module extraction have been proposed in the
literature~\cite{Grau:2008:MRO:1622655.1622664,DBLP:journals/ai/KonevL0W13,
  DBLP:journals/jair/RomeroKGH16,Stuckenschmidt:2009:MOC:1571642,Kontchakov20101093, DelVescovo:2011:MSO:2283696.2283770}. Important
technical tools in modularity are various algorithms related to
\emph{conservative extensions}, \emph{inseparability}, \emph{uniform
  interpolation}, and \emph{forgetting}, which have been studied in DLs
but are also classic problems in other areas of logic (see,
e.g.,~\cite{Lutz:2011:FUI:2283516.2283563,DBLP:journals/jair/BaaderBL16,DBLP:journals/ai/BotoevaKRWZ16}). Roughly
speaking, a fragment $\mathcal{O}_1$ of an ontology $\mathcal{O}_2$ is
called a \emph{module} if the terms defined in $\mathcal{O}_1$ have
precisely the meaning they have in the larger ontology
$\mathcal{O}_2$. More precisely, one usually requires that
$\mathcal{O}_1$ and $\mathcal{O}_2$ agree on the entailment of
inclusions over a given signature $\Sigma$. In this way, an
application whose scope is limited to the entities in $\Sigma$ can
safely use $\mathcal{O}_1$ instead of the full $\mathcal{O}_2$. The
difference between the mentioned works and our work is that
focusing---unlike module extraction---will in general change the
meaning of terms. To see this, consider a (rather trivial) disaster
management ontology
$\mathcal{O}=\{\mathsf{Disaster}\,{\equiv}\,\mathsf{Flood} {\sqcup}
\mathsf{Drought}\}$. Suppose we want to focus on floods, and thus we
naturally expect a focusing configuration to lead to the entailment of
$\mathsf{Disaster} \equiv \mathsf{Flood} $.  This can be achieved via
a focusing solution
$\mcF=(\{\mathsf{Flood}\},\{\mathsf{Flood}\},\{\mathsf{Drought}\},\emptyset)$. Since
the desired entailment does not hold in $\Oo$, module extraction
cannot help us in this case. As a method that does change the meaning
of terms, we mention \cite{DBLP:conf/semweb/ChenLMW17}, where the
authors consider the problem of selecting $n$ inclusions from an input
ontology that preserve as much entailments as possible.

\vspace*{-.3em}
\subsubsection*{Closed Predicates} Combining closed-world and
open-world reasoning is a recognized challenge both in database and AI
research, and has received significant attention in the literature.
The use of \emph{closed predicates} is a particular way to overcome
the
challenge~\cite{DBLP:conf/ijcai/LutzSW15,DBLP:conf/ijcai/LutzSW13,DBLP:journals/entcs/FranconiIS11}. Note
that Definition~\ref{def:closed-queries} generalizes the idea of
\emph{closed predicates} in DLs.  Closing extensions of predicates is
similar in spirit to
\emph{circumscription}~\cite{DBLP:journals/ai/McCarthy80}, but instead
of minimizing the inference of new tuples in selected predicates, such
inferences are prohibited altogether. Circumscription has been studied
both for expressive and lightweight
DLs~\cite{DBLP:journals/jair/BonattiLW09,DBLP:journals/jair/BonattiFS11}. Note
that closed predicates, or other kinds of statements to assert
information completeness have also been studied in databases 
(see, e.g.,
\cite{DBLP:conf/lics/BenediktBCP16,Fan:2010:RIC:1862919.1862924,Abiteboul:1998:CAQ:275487.275516}). In
particular, \cite{Fan:2010:RIC:1862919.1862924} studies completeness
assertions made using queries, and explores reasoning about databases and queries in the presence of such~assertions. 

\vspace*{-.3em}
\subsubsection*{Query Emptiness} The nullability problem studied in
Section~\ref{sec:recogn} is closely related to the \emph{query
  emptiness} problem studied in~\cite{DBLP:journals/jair/BaaderBL16},
and the so-called \emph{schema-level positive query implication}
($\exists$PQI) problem studied
in~\cite{DBLP:conf/lics/BenediktBCP16}. For Boolean
queries, the 3 problems effectively collapse (here and in the
mentioned papers, the problems are studied for different
languages). For non-Boolean queries, there is a slight
divergence. Consider the ontology $\Oo=\{\{c\}\ISA \exists r.A\}$, the
instance query $A(x)$, and suppose the signature of legal databases is
$\Sigma=\emptyset$. In the sense of
\cite{DBLP:journals/jair/BaaderBL16, DBLP:conf/lics/BenediktBCP16},
this yields a positive instance, i.e., there is a database (the empty database) in which
the certain answer to $A(x)$ is empty. In our setting, this is a
negative instance of the nullability problem (the extension of $A$
always has an element, still we cannot identify it via a constant). Note
that the undecidability result for nullability with $\ELIbot$
ontologies already holds for simple Boolean CQs of the form
$\exists x\, A(x)$ (see Appendix). This contributes to the study
initiated in~\cite{DBLP:conf/lics/BenediktBCP16}, which showed
undecidability results for disjunctive linear tuple generating
dependencies (TGDs) and linear TGDs with constants. Since $\ELIbot$
can be translated into guarded TGDs without
constants (and using only predicates of arity at most two) 
but extended with constraints, this provides a new class of
constraints for which $\exists$PQI~is~undecidable.

\vspace*{-.3em}
 \subsubsection*{Finite Model Reasoning} Reasoning about finite models of
 DL ontologies has received some (albeit
 limited) attention~\cite{thesis-1996,DBLP:journals/iandc/LutzST05,DBLP:conf/esws/Rosati08,DBLP:conf/kr/GarciaLS14,Rudolph16,GogaczIM18}. Since
  DLs are closely related to the  fragments  $\mathcal{C}^{1}$ and $\mathcal{C}^{2}$ of first-order logic with counting quantifiers, the work by
 Pratt-Hartmann on the complexity of finite and unrestricted
 satisfiability in these logics is particularly
 relevant~\cite{DBLP:journals/jolli/Pratt-Hartmann05,DBLP:journals/bsl/Pratt-Hartmann08}.  The
 inequations we use in Section~\ref{sec:emptiness} are inspired by
 \cite{DBLP:journals/iandc/LutzST05}, but we use some tricks
 from~\cite{DBLP:journals/jolli/Pratt-Hartmann05} (in addition to our
 own) to deal with mixed satisfiability.

\section{Discussion}
\label{sec:discussion}

We have introduced \emph{focusing}, which makes it possible to reuse the
knowledge in an ontology as a basis for the on-demand design of data-centric
applications. 
The selected complexity results we have provided are not meant to paint a
complexity landscape, or to identify the best formalisms to be used for
focusing. Rather, they constitute a preliminary study of the limits and
possibilities of the focusing framework, and many questions
remain to be answered.

We have briefly mentioned as an open issue the design of strategies to
find preferred focusing solutions, which also calls for  suitable
ways to capture preferences of the designer. 
There are also many challenges related to reasoning about the
queries in focusing specifications: e.g.\,to derive additional queries that
may be closed or fixed, or to use the closed and fixed assumptions for query
optimization. 

The specification of queries that are complete or fixed is  a way to add knowledge
to the ontology, and this knowledge could be leveraged in many ways. 
For example, it is well known that CQs with negation are undecidable in the
presence of ontologies~\cite{DBLP:journals/ws/Gutierrez-Basulto15}, but over 
our enriched ontologies, non-trivial queries with negation 
admit algorithms. This applies, e.g., to CQs with negation where each variable that occurs in a negative atom
also occurs  in a positive  atom over a closed predicate. 
Moreover, one could imagine a scenario where one may use queries with
negation to specify the scope of the desired system, and the focusing engine
could take that into account and search for focusing solutions where suitable
predicates are closed as to guarantee decidability of query answering. 

We would also like to understand whether in some cases, fixing predicates has
a significant effect on the complexity of reasoning, for example, if the more
specific ontology that is implicitly built by fixing is in a logic with lower
complexity (this  could happen, for example, if we fixed a predicate from every
disjunction, making the specific part of the ontology effectively Horn). 
Are there ways to decide whether this is the case,  and if so, can this be
effectively leveraged by algorithms?
 



%% file: appendix-emptiness.tex
\section{Missing Proofs of Section~\ref{sec:recogn}}

\newcommand{\up}{\mathsf{up}}
\newcommand{\head}{\mathsf{Head}}
\newcommand{\nohead}{\mathsf{NoHead}}

\newcommand{\noheadright}{\mathsf{NoHeadRight}}
\newcommand{\noheadleft}{\mathsf{NoHeadLeft}}

\newcommand{\myright}{\mathsf{right}}
\newcommand{\orig}{\mathsf{Orig}}

We now argue that allowing roles names to be used to specify
completeness queries leads to undecidability of the nullability problem. This holds for the DL
$\ELIbot$, which is a well-known \emph{Horn} DL. 
\begin{theorem}\label{thm:undecidablity}
  $\qempty(\parTELIbot,\parCAQ,\parQIQ)$ is undecidable. Moreover,
  $\qempty(\parTELIbot,\parCAQ,\parQCQ)$ is undecidable even for
  Boolean CQs of the form $\exists x.A(x)$, where $A$ is a concept
  name.
\end{theorem}
\begin{proof}
  We provide a reduction from the \emph{halting problem} for
  (deterministic) \emph{Turing Machines (TMs)}. Assume a TM
  \[M=(\Gamma,Q,q_0,q_{\mathit{yes}},q_{\mathit{no}},\delta)\] with
  the alphabet $\Gamma=\{0,1\}$, the state set $Q$, the initial state
  $q_0\in Q$, the accepting state $q_{\mathit{yes}}\in Q$, the
  rejecting state $q_{\mathit{no}}\in Q$, and the transition function
  $\delta:Q\setminus\{q_{\mathit{yes}},q_{\mathit{no}}\} \times
  \Gamma\cup\{b\} \rightarrow Q\times \Gamma\cup\{b\} \times
  \{+1,-1\}$. The problem is to decide whether $M$ terminates on the
  empty word $\epsilon$; that is, starting from the one-way infinite
  tape that contains in every cell only the blank symbol $b$.  We
  assume that the machine never comes back to the initial
  state $q_0$, and never attempts to move to the left of the initial
  position of the head. We will now construct an instance
  $(\Oo,\Sigma, \Qcwa,q)$ of $\qempty(\parTELIbot,\parCAQ,\parQIQ)$
  such that $(\Oo,\Sigma, \Qcwa,q)$ is a negative instance iff the
  machine $M$ terminates on $\epsilon$. In the construction of $\Oo$,
  we use the following symbols to encode runs of TMs:
\begin{itemize}
\item We use role names $\up$ and $\myright$ to encode a grid on which
  the computation of $M$ can be represented. 
  \item We use the concept name $\head$ to indicate the position of the read-write head of $M$.
  \item We use the concept name $\noheadright$ and $\noheadleft$  to indicate that the
    read-write head is not here or to the right (resp.\,left)  of a given tape cell.

\item We use the concept name $Q_q$ for every state $q\in Q$.
\item We use the concept names $S_{0},S_{1},S_{b}$ to indicate the
  content of a tape cell.
\item We use the concept names $X,Y_1, Y_2, Z$ to perform a trick, which will become clear later.
\item We use an auxiliary role name $r$.
\end{itemize}

Our next goal is to obtain the ontology $\Oo$. Its construction is
split into two parts. In the first part, we simply define inclusions
that simulate the computation of $M$ on a grid given by the roles
$\up$ and $\myright$. The (extension of the) role  $\myright$
corresponds to the tape of $M$, while the (extension of the) role
$\up$ corresponds to time. To make the presentation easier, we allow
infinite instances over $\Qcwa$; modifying the below reduction to fit
precisely the definition of nullability (which requires finite
instances) is not too dificult. The following inclusions of $\Oo$ are
self-explanatory (assuming that $\up$ and $\myright$ form a grid,
which is the tricky part here):   

\begin{enumerate}[(A)]
  \itemsep.8em

\item $\top \ISA  \exists r. (
  Q_{q_0}\sqcap \head)$
\item $Q_q \ISA (\exists \up .\top) \sqcap (\exists \myright .\top)$ for all $q\in Q\setminus\{q_{\mathit{yes}},q_{\mathit{no}}\}  $

\item $ \exists \myright^{-}.\head \ISA \noheadright$

 \item $ \exists \myright.\head \ISA \noheadleft$

 \item $\noheadright \ISA \nohead $

 \item $\noheadleft  \ISA \nohead $

 \item $\nohead \sqcap \head \ISA \bot$  

\item $\exists \myright^-.\noheadright \ISA \noheadright $
  
\item $\exists \myright.\noheadright \ISA \noheadleft $

    
    \item $\exists \myright. Q_q \ISA  Q_q $ for all $q\in Q$

       \item $\exists \myright^-. Q_q \ISA  Q_q $ for all $q\in Q$

       \item $Q_{q_{0}}\ISA S_{b}$

         \item $Q_q\sqcap Q_s\ISA \bot$ for all $q,s\in Q$

    \item $\exists \up^{-}. (S_{\sigma}\sqcap \nohead) \ISA S_{\sigma}$ for all $\sigma\in \{0,1,b\}$

    \item $\exists \up^{-}. (S_{\sigma} \sqcap Q_q \sqcap \head)\ISA S_{\sigma'}\sqcap Q_{q'} $  for all
      $(q,\sigma)\in Q\setminus\{q_{\mathit{yes}},q_{\mathit{no}}\}
      \times \Gamma\cup\{b\} $ such that  $\delta(q,\sigma)=(q',\sigma',D)$

       \item $\exists \myright^-.(\exists \up^{-}. (S_{\sigma} \sqcap Q_q \sqcap \head))\ISA     \head$  for all
      $(q,\sigma)\in Q\setminus\{q_{\mathit{yes}},q_{\mathit{no}}\}
      \times \Gamma\cup\{b\} $ such that  $\delta(q,\sigma)=(q',\sigma',+1)$

       \item $\exists \myright.(\exists \up^{-}. (S_{\sigma} \sqcap Q_q \sqcap \head))\ISA     \head$  for all
      $(q,\sigma)\in Q\setminus\{q_{\mathit{yes}},q_{\mathit{no}}\}
      \times \Gamma\cup\{b\} $ such that  $\delta(q,\sigma)=(q',\sigma',-1)$

      \end{enumerate}

      We let $\Qcwa$ be the set of concepts and roles mentioned in the
      above inclusions, except $r$. Consider an instance $\mcI$ over
      $\Qcwa$. We say $\mcI$ is \emph{grid-like} in case 
       $\{\up (c,d),\myright(d,e)\}\subseteq\mcI$ and $\{\myright(c,d'),\up(d',e')\}\subseteq\mcI$ imply
      $e=e'$. It is easy to see that the above encoding is
      correct on grid-like instances: (a) if $M$ halts, then there is a
      grid-like instance $\mcI$ over $\Qcwa$ such that $\mcI$ satisfies
      every inclusion above, and (b) if there exists a grid-like
      instance $\mcI$ over $\Qcwa$ that satisfies all the above
      inclusions, then $M$ halts. We need a tool to ensure that our instances are grid-like. For this we use further
      inclusions, which will involve predicates that are not part of
      $\Qcwa$. These predicates are the concept names $X,Y_1,Y_2,Z$, and the
      role name $r$. We add the following inclusions~to~$\Oo$:

\begin{enumerate}[(1)]
    \item  $\top \ISA \exists r. (X \sqcap (\exists \up.\exists \myright. Y_1)\sqcap (\exists \myright.\exists \up. Y_2))$

 

\item $ Y_1\sqcap Y_2 \ISA Z$

  \end{enumerate}
  The inclusions above can be explained as follows. We use $X,Y_1,Y_2$
  as ``flags''. The inclusion (1)  says that (i) $X$ needs to be
  arbitrarily placed somewhere in the candidate to be a grid-like
  structure, (ii)  $Y_1$ needs to placed at some up-right neighbor of
  (the placement of) $X$, and  (iii)  $Y_2$ needs to placed at some
  right-up neighbor of (the placement of) $X$.
  The inclusion
  (2) just says that if a location is flagged with $Y_1$ and $Y_2$, then
  also $Z$ is placed there.

  Finally, the query $q$ is $q=Z(x)$, or alternatively
  $q=\exists x.Z(x)$.

  It is not difficult to see that the above is a correct
  reduction. Indeed, if $M$ halts, then we can easily translate a
  terminating run of $M$ into a grid-like instance $\mcI$ over $\Qcwa$
  such that $Z^{\mcJ}\neq \emptyset$ for all models $\mcJ$ of $\Oo$
  that agree with $\mcI$ on the extensions of the predicates in
  $\Qcwa$. 
  For the other direction, assume we have an instance $\mcI$ over
  $\Qcwa$ that forces the extension of $Z$ to be non-empty in all
  relevant models of $\Oo$.  This implies that $\mcI$ is grid-like,
  and thus a terminating run of $M$ can be extracted. To see this a
  bit more formally, suppose that $\mcI$ is not grid-like. Then there
  exists some
  $\{\up (c,d),\myright(d,e),\myright(c,d'),\up(d',e')\}\subseteq\mcI$
  such that $e\neq e'$. Extend $\mcI$ to a model of $\Oo$ by adding
  $X(c),Y_1(e),Y_1(e')$ (and additionally $r(t,c)$ for all constants
  $t$). One arrives at a contradiction: the extension of $Z$ in this instance is empty.
%
%
%
%
%
%
%
%
%
%
%
\end{proof}

Next, we show that more radical restriction of the logic can restore
the decidability of the nullability problem with closed roles.  

\begin{theorem}
$\qempty(\parTLITEOBOOL,\parCAQ,\parQAQ)$ is in  $\conexptime^\np$.
 \end{theorem}

 \begin{proof}
We give a polynomial reduction to
$\qempty(\parTALCOI,\parCIQ,\parQAQ)$. Consider a $\DLLITEOBOOL$
ontology $\Oo$, a signature $\Sigma \subseteq\conceptnames(\Oo)
\cup\rolenames(\Oo)$, a set of queries $\Qcwa \subseteq \AQ$, and a
query $q \in \AQ$. Let $\Sigma_{\mn{R}}$ be the set of roles $r$ such
that $r\in\Qcwa$ or $r^-\in \Qcwa$.  For each $r\in \Sigma_{\mn{R}}$,
introduce a fresh concept name $A_r$ axiomatized as $A_r \equiv
\exists r. \top$; that is, $A_r$ is intended to contain origins of all
$r$-edges. We also include axiom $ A_r \ISA \exists r. A_{r^-}$, which
says that each origin of $r$ has an origin of $r^-$ among its
$r$-successors. The resulting ontology $\Oo'$ is not in
$\DLLITEOBOOL$, but it is in $\ALCOI$. Let $\Qcwa'$ be obtained by
removing from $\Qcwa$ all role queries, and adding all queries of the
form $A_r(x)$ for $r\in \Sigma_{\mn{R}}$. Let $\Sigma' = \Sigma \cup
\{A_r \mid r\in \Sigma\cap\Sigma_{\mn{R}}\}$.
We claim that $q$ is
nullable wrt.~$\Oo$, $\Sigma$, and $\Qcwa$ iff it is nullable
wrt.~$\Oo'$, $\Sigma'$, and $\Qcwa'$.

For any instance $\Mm$ over the signature of $\Oo$, let $\widehat \Mm$
be obtained by extending $\Mm$ minimally so that
$(A_r)^{\widehat{\Mm}} = (\exists r.\top)^{\Mm}$. Note that if
$\Mm\models \Oo$, then $\widehat{\Mm} \models \Oo'$.

Suppose there is $\Ii$ over $\Sigma$ such that $\rest(\Ii,\Oo,\Qcwa)\neq \emptyset$
and $\answer{\Jj}{q} \neq \emptyset$ for all $\Jj \in
\rest(\Ii,\Oo,\Qcwa)$.  We claim that $\Ii' =\widehat{\Ii}$ is a
counter-witness for the nullability of $q$ wrt. $\Oo'$, $\Sigma'$, and
$\Qcwa'$. If $\Jj \in \rest(\Ii,\Oo,\Qcwa)$ then $\widehat{\Jj} \in
\rest(\Ii',\Oo',\Qcwa')$, so $\rest(\Ii',\Oo',\Qcwa')\neq\emptyset$. Suppose $\answer{\Mm'}{q}=\emptyset$ for
some $\Mm' \in \rest(\Ii',\Oo',\Qcwa')$. Because each $A_r$ is closed,
only elements that had an outgoing $r$-edge in $\Ii$ will have new
outgoing $r$-edges in $\Mm'$. Hence, after removing all new $r$-edges,
$\Mm'$ is still a model of $\Oo'$. Consequently, we can assume
wlog.~that there are no new $r$-edges in $\Mm'$ for all $r\in
\Sigma_{\mn{R}}$; that is, $r^{\Mm'} = r^{\Ii}$ for all $r\in
\Sigma_{\mn{R}}$. Let $\Mm$ be the restriction of $\Mm'$ to the
signature of $\Oo$. Clearly $\Mm \in \rest(\Oo, \Ii, \Qcwa)$ and
$\answer{\Mm}{q}=\emptyset$.

Conversely, suppose there is $\Ii'$ over $\Sigma'$ such that $\rest(\Oo', \Ii',\Qcwa')\neq
\emptyset$ and $\answer{\Jj'}{q} \neq \emptyset$ for all $\Jj' \in
\rest(\Oo', \Ii',\Qcwa')$. Take some $\Jj' \in \rest(\Oo',
\Ii',\Qcwa')$. Let $\Ii''$ be obtained from $\Ii'$ by including all
$r$-edges present in $\Jj'$ for all $r\in\Sigma_{\mn{R}}\cap \Sigma'$. Note that
\[\Jj' \in \rest(\Oo', \Ii'',\Qcwa') \subseteq \rest(\Oo',
\Ii',\Qcwa')\,.\] Hence, wlog. we can assume that all
these edges are already present in $\Jj'$; that is, $r^{\Ii'} =
r^{\Jj'}$ for all $r\in\Sigma_{\mn{R}}\cap \Sigma'$. Let $\Ii$ and
$\Jj$ be the restrictions to the signature of 
$\Oo$ of the instances $\Ii'$ and $\Jj'$. It holds that
$\Jj\in\rest(\Oo,\Ii,\Qcwa)$. Suppose that there exists $\Mm\in
\rest(\Oo,\Ii,\Qcwa)$ with $\answer{\Mm}{q} = \emptyset$. Let $\Mm'
= \widehat{\Mm}$. We have $\Mm'\in \rest(\Oo',\Ii',\Qcwa')$ and
$\answer{\Mm'}{q} = \emptyset$. 
%
\end{proof}







\section{Missing Proofs of Section~\ref{sec:entailment}}


\begin{proof}[Proof of Lemma~\ref{lem:unravelling}]
The second condition trivially implies the first one. For the converse
implication, assume there is some counter model $\Jj$ for $q$.  Using
a standard unraveling technique, we will inductively construct from
$\Jj$ a tree extension $\Tt$ of $\Ii$ together with a homomorphism $g$
from $\Tt$ to $\Jj$, and show that $\Tt \in \rest(\dlkb,\Ii,\Qcwa)$.
Let $\Tt_0$ be the subinstance of $\Jj$ induced by $\adom{\Ii}$ and
let $g_0$ be the identity homomorphism from $\Tt_0$ to $\Jj$.  Assume
$\Tt_i$ and $g_i$ are defined. To define $\Tt_{i+1}$ and $g_{i+1}$,
for each element $d \in \adom{\Tt_i}$ and each axiom $K \sqsubseteq
\exists r. B$ that is not yet satisfied in $d$, consider an element $e
\in B^{\Jj}$ such that $(g_i(d), e) \in r^\Jj$.  If $e \in\adom{\Ii}$
then $\Tt_{i+1}$ extends $\Tt_i$ by adding an $s$-edge from $d$ to $e$
for every super role $s$ of $r$.  Otherwise, we add a fresh copy $e'$
of $e$ and an $s$-edge from $d$ to $e'$, for every super role $s$ of
$r$, and extend $g_i$ to $g_{i+1}$ by setting $g_{i+1}(e')=e$.
Clearly, $\Tt$ is a tree-extension of $\Ii$ such that $\Tt \models
\dlkb$. Because $\Ii \subseteq \Tt$ and $\Tt$ maps homomorphically to
$\Jj$ we have $\answer{\Ii}{q'} \subseteq \answer{\Tt}{q'} \subseteq \answer{\Jj}{q'}$, for every
$q' \in \Qcwa$. Because $\answer{\Jj}{q'} = \answer{\Ii}{q'}$, it
follows that $\answer{\Tt}{q'} = \answer{\Ii}{q'}$. An analogous
argument proves that $\Tt \nvDash q$.
\end{proof}

\section*{Missing Proofs of Section~\ref{sec:emptiness}}

\begin{proof}[Proof of Proposition~\ref{prop:mixed-reduction}]
  Assume $(\varphi,\mcF)\in\recognition(\parTL, \parCAQ, \parDany)$,
  where $\mcF= (\Sigma,\Qcwa,\Qdet,\emptyset)$. Then $(\varphi,\mcF)$
  is a positive instance of $\emptiness(\parTL,\parCAQ,\parDany)$ iff
  $(\varphi,\Qcwa)$ is a positive instance of  $\mixedunsat(\parTL)$.

  Assume an instance $(\varphi,\Sigma)$ of
  $\mixedunsat(\parTL)$. Then
  $(\varphi,(\Sigma,\Sigma,\emptyset,\emptyset))$ is trivially a
  positive instance of $\recognition(\parTL,
  \parCAQ)$, and further
  $(\varphi,\Sigma)$ is a positive instance of
  $\mixedunsat(\parTL)$ iff
  $(\varphi,(\Sigma,\Sigma,\emptyset,\emptyset))$ is a positive
  instance of $\emptiness(\parTL, \parCAQ)$.
\end{proof}

\begin{proof}[Proof of Theorem~\ref{thm:mosaic}]
  \emph{(from (1) to (2))} Assume a model $\mcI$ of $\dlkb$ where all
  the concept names in $\Sigma$ have a finite extension. We can assume
  that $\adom{\Ii}$ is finite or countably infinite. We define a
  mosaic $N$ for $(\dlkb,\Sigma)$.  For an element
  $e\in\adom{\Ii}$, let
  \[\mathit{ut}(e)=\{\top\}\cup \{B\in\simpleconcepts(\dlkb)\mid e\in
  B^{\mcI} \}\,.\] For a pair of elements $e,e'\in\adom{\Ii}$, let
  \[\mathit{bt}(e,e')=\{r\in \roles(\dlkb) \mid (e,e')\in
  r^{\mcI}\}\,.\] Assume an element $e\in \adom{\Ii}$. Since $\mcI$ is
  a model of $\dlkb$, for all $\alpha = A\ISA \exists r. B\in \dlkb$
  with $A\in \mathit{ut}(e)$, there exists an element $e_{\alpha}$
  such that $(e,e_{\alpha})\in r^{\mcI}$ and $e_{\alpha}\in
  B^{\mcI}$. Let $e_{\alpha_1},\ldots,e_{\alpha_n}$ be the set such
  elements for each axiom $\alpha_i$ of the form $ A\ISA \exists r. B$
  occurring in $\dlkb$. Let $F$ be the set of elements
  $e'\in \adom{\Ii}$ such that $(e,e')\in r^{\mcI}$ for some
  $\func(r)\in \dlkb$. Let
  $S=\{e_{\alpha_1},\ldots,e_{\alpha_n}\}\cup F$. We next define the
  tile extracted for the element $e$ as follows:
  \[ \mathit{tile}(e)=\big
    (\mathit{ut}(e),\{(\mathit{bt}(e,e'),\mathit{ut}(e'))\mid e'\in
    S\}\big ).\] Note that for a given element $e$, there can be
  several ways to extract a tile $\mathit{tile}(e)$, because
  $e_{\alpha_1},\ldots,e_{\alpha_n}$ can be chosen in multiple
  ways. It is not difficult to check that $\mathit{tile}(e)$ is indeed
  a proper tile. Let $N $ be the function such that, for each
  $\tau \in \tiles(\dlkb)$,
  \[N(\tau)=|\{e\in \adom{\Ii}\mid \tau = \mathit{tile}(e) \}|\,.\] Note
  that $N(\tau)=\alnull$ in case $\{e\in \adom{\Ii}\mid \tau =
  \mathit{tile}(e) \} $ is infinite. It is not difficult to check that
  $N$ is a proper mosaic. 
  
  \smallskip
  
  \emph{(from (2) to (1))} Assume there exists a mosaic $N$ for
  $(\dlkb,\Sigma)$. We show how to construct a model $\mcI$ of $\dlkb$
  such that all predicates from $\Sigma$ have a finite extension. 
Let $\Delta = \big\{ (\tau,i)\in \tiles(\dlkb)\times\mathbb{N} \mid 0  
  < i\leq N(\tau) \big\}.$
For each $(\tau,i) \in \Delta$ we take a designated constant, which in a
slight abuse of notation,  we denote simply as $(\tau,i)$.  

  For simple concepts
  $B\in \simpleconcepts(\dlkb)\setminus\{\top,\bot\}$, we let
  $B^{\mcI}=\{((T,\rho),i)\in\Delta\mid B\in T\}$. It now
  remains to define the extensions of roles. First, we deal with roles
  $r$ such that both $\func(r)\in \dlkb$ and $\func(r^-)\in \dlkb$,
  and for this we employ the inequation (\ref{item:one-way-matching})
  in two directions. Intuitively, the role $r$ needs to connect an
  element $e $ with some ``unary type'' $T$ to some element $e'$ with
  some unary type $T'$ via some ``binary type'' $R$ that contains
  $r$. For each such triple $T,T',R$, due to the equation
  (\ref{item:one-way-matching}), there is a bijection $f$ between the
  sets $A$ and $B$, where
  \[A=\big\{((T,\rho),i)\in \Delta\mid (R,T')\in \rho \big\}
    \qquad \text{and} \qquad 
  B=\big\{((T',\rho'),i)\in \Delta\mid (R^{-},T)\in \rho' \big\}\,.\] For
  every $e\in A$, let $(e,f(e))\in s^{\mcI}$ iff $s\in R$. Now we
  consider each role $r$ such that $\func(r^{-})\in \dlkb$ but
  $\func(r)\not\in \dlkb$.  Again, the role $r$ connects an element
  $e$ with some unary type $T$ to some element $e'$ with some unary
  type $T'$ via some binary type $R$ that contains $r$. For each such
  triple $T,T',R$, due to the equation (\ref{item:one-way-matching}),
  there is an injective function $f$ from the set $A$ to the set $B$, where
  \[A=\big\{((T,\rho),i)\in \Delta\mid (R,T')\in \rho \big\}
    \qquad \text{and}  \qquad
    B=\big\{((T',\rho'),i)\in \Delta\mid (R^{-},T)\in \rho' \big\}\,.\]
  For every $e\in A$, let $(e,f(e))\in s^{\mcI}$ iff $s\in R$. Remains
  to define the remaining role extensions.  
  In particular, pick a pair $T,T'\in \mathsf{Types}(\dlkb)$ and a set
  $R\subseteq \roles(\dlkb)$ that does not contain any functional or
  inverse functional roles of $\dlkb$. Let
  \[A=\big\{((T,\rho),i)\in \Delta\mid (R,T')\in \rho \big\}\,.\] If
  $A=\emptyset$, then do nothing. If $A\neq\emptyset$, then we know by
  condition (\ref{item:successor}) that this is some $\rho'$ such that
  $e'=((T',\rho),1)\in \Delta$. For every $e\in A$, we let
  $(e,e')\in s^{\mcI}$ iff $s\in R$.
  This finishes the construction of $\mcI$. It is not too difficult to
  check that $\mcI$ is indeed a model of $\dlkb $ such that all
  predicates from $\Sigma$ have a finite extension.
\end{proof}

  \begin{proof}[Proof of Theorem~\ref{thm:enriched-ineq}]
  Given the above reduction, in order to prove the \nexptime upper
  bound, it remains to argue that checking the existence of a solution
  $S$ to any $(V,\mathcal{E},F,I)$ is feasible in non-deterministic
  polynomial time. For the \exptime upper bound, it suffices to show
  that the existence of $S$ can be decided in polynomial time in case
  $\mathcal{E}$ contains only positive inequations. Consider a system
  \[\mathcal{H}^*=(V^*,\mathcal{E}^*,V^*,I^*)\,.\] Checking the existence
  of a solution to $\mathcal{H}^*$ is solvable in non-deterministic
  polynomial time, because in non-deterministic polynomial time we can
  build an ordinary inequation system
  $(V^*,\mathcal{E}^*\cup\hat{\mathcal{E}})$ such that $\mathcal{H}^*$
  has a solution iff $(V^*,\mathcal{E}^*\cup\hat{\mathcal{E}})$ has a
  solution. The inequations $\hat{\mathcal{E}}$ are obtained from the
  implications $y_1+\cdots +y_m>0 \Rightarrow x_1+\cdots + x_n>0$ in
  $I^*$ by non-deterministically adding one of the inequations
  $y_1+\cdots +y_m = 0$, $x_1 > 0,\ldots, x_n > 0 $ to
  $\hat{\mathcal{E}}$. If $\mathcal{E}^*$ is a set of positive
  inequations, then checking the existence of a solution to
  $\mathcal{H}^*$ is feasible in polynomial time.

  \medskip We next argue that a given $(V,\mathcal{E},F,I)$ can be
  transformed in polynomial time into an enriched inequation system
  $(V',\mathcal{E}',V',I')$, while preserving the existence 
  of a solution. The resulting system still contains implications, but
  its solutions must be over $\nat$. We let $V'$ be obtained by adding
  to $V$ a fresh variable $x^{\infty}$ for every $x\in V$. The
  inequations in $\mathcal{E}'$ are obtained from the inequations in
  $\mathcal{E}$ by replacing every summand $b_i\cdot y_i$ that occurs
  on the right-hand-side of ``$\leq $'' of an inequation by the
  summand $b_i\cdot y_i + y_i^{\infty}$.  Second, for every $x\in F$,
  we add $x^{\infty}=0$. For example, an inequation
  $2\cdot x \leq 2\cdot y+z$ is replaced by
  $2\cdot x \leq 2\cdot y + y^{\infty}+z + z^{\infty}$.  The set $I'$
  is obtained from $I$ in two steps. First, we replace in every
  implication every summand $b_i\cdot y_i$ by the summand
  $b_i\cdot y_i + y_i^{\infty}$. Second, we add to $I'$ the
  implication
\begin{equation}
      x_1^{\infty}+\cdots+ x_n^{\infty} >0  \Rightarrow  y_1^{\infty}+\cdots+y_m^{\infty} >0
  \end{equation}
  for every inequation of the form \eqref{eq:linear} in
  $\mathcal{E}$. It is not difficult to see that a solution $S$ to
  $\mathcal{H}$ can transformed into a solution $S'$ for
  $\mathcal{H}'$. Let $S$ be a solution to $\mathcal{H}$, and let
  $B = 1 + c + \sum_{x\in F} c \cdot S(x)$, where $c$ is the maximal
  value among the constants that appear in $\mathcal{E}$. We let $S'$ be defined as follows:
   \begin{equation}
  S'(x) = \begin{dcases}
    S(x), & x\in V \land   S(x) \in \nat \\
    0, & x\in V \land   S(x)=\alnull \\
     0, & x=y^{\infty} \land S(y)\in \nat \\
     B, & x=y^{\infty} \land S(y) =\alnull
  \end{dcases}
\end{equation}
It is easy to see that $S'$ is a solution to $\mathcal{H'}$. For the
 other direction, assume $S'$ is an arbitrary solution to
 $\mathcal{H'}$. We define $S:V'\rightarrow \nat^{*}$ as follows:
   \begin{equation}
  S(x) = \begin{dcases}
    S'(x), &   S'(x^{\infty}) = 0 \\
    \alnull, & S'(x^{\infty}) > 0 
  \end{dcases}
\end{equation}
It is not difficult to see that that $S$ is a solution to $\mathcal{H}$.
      \end{proof}

\subsection*{Mixed Satisfiability in $\DLLITEBOOL$}

Let us assume a $\DLLITEBOOL$ ontology $\dlkb$, and let $\Sigma$ be a set of
concept names we want to keep finite. We
   use $\ISA^{*}_{\dlkb}$ to denote the reflexive transitive closure of
   $\{(r,s),(r^-,s^-)\mid r\ISA s\in\dlkb\}$. 
A set $T$ of
 simple concepts is an \emph{$r$-sink} (w.r.t.~$\dlkb$) if $A\in T$
for all $\top \ISA \forall s. A \in \dlkb$ such that
$r\ISA_{\dlkb}^{*} s$.  A \emph{tile} for $\dlkb$ is a pair
$\tau=(T,R)$, where $T$ is a set of simple concepts, and
$R\subseteq \roles(\dlkb)
$
such that
\begin{enumerate}
  \itemsep0.6em
  \item for all  $B_1\sqcap B_2 \ISA B_3\sqcup B_4 \in \dlkb$,  
     $ \{B_1,B_2\}\subseteq T$ implies 
     $ \{B_3,B_4\}\cap T\neq\emptyset$;
  
 \item for all $A\ISA \exists r. \top \in \dlkb$, $A\in T$ implies  $r\in R$ 
   \item  $T$ is an $r$-sink for all $r^{-}\in R$
\end{enumerate}
We use $\tiles(\dlkb)$ to denote the set of all tiles for $\dlkb$. A
tile $(T,R)$ is an $r$-sink, if $T$ is an $r$-sink. We let $\mathrm{root}((T,R))=T $. We now want to use
mosaics to check the existence of a model for $\dlkb$ that keeps finite
the extensions of the predicates in $\Sigma$.

\begin{definition}\label{def:dllite-mosaic}
A \emph{mosaic} for
$(\dlkb,\Sigma)$ is a function
$N:\tiles(\dlkb)\rightarrow \mathbb{N}\cup \{\alnull\} $ satisfying the
following:
\begin{enumerate}
  \itemsep0.6em
   \item There is a tile $\tau\in \tiles(\dlkb)$ such that $N(\tau)>0$.

     \item For every nominal $\{o\}$ that appears in $\dlkb$, there is a tile $(T,R)$ such that  $\{o\}\in T$ and  $N((T,R))>0$.

   \item For every pair of tiles $\tau_1,\tau_2\in \tiles(\dlkb)$ and every nominal $\{o\}$ that appears in $\dlkb$, $\{o\}\in \mathrm{root}(\tau_1)\cap  \mathrm{root}(\tau_2)$,  $N(\tau_1)>0$ and $N(\tau_2)>0$ imply $\tau_1=\tau_2$ and $N(\tau_1)=1$.

   \item For every $A\in \Sigma$ and every tile $\tau\in \tiles(\dlkb)$ with $A\in \mathrm{root}(\tau) $, we have $N(\tau)\neq \alnull$.

 \item For all tiles $\tau= (T,R) $ and all $r\in R$, if $N(\tau)> 0$,
   then there exists a tile $\tau'$ such that $N(\tau')>0$ and $\tau'$
   is an $r$-sink.

    \item For every $r$ with $\func(r),\func(r^-)\in \dlkb$, we have
    \[\sum_{\substack{(T,R)\in \tiles(\dlkb) \,\land \\
                  r\in R}}N((T,R)) \quad =\quad  \sum_{\substack{(T,R)\in \tiles(\dlkb)\,\land \\
                  r^{-}\in R} }N((T,R)) \;.
    \]

   \item For every $r$ with $\func(r^-)\in \dlkb$ but $\func(r)\not\in \dlkb$, we have
    \[\sum_{\substack{(T,R)\in \tiles(\dlkb)\,\land \\
                  r\in R}}N((T,R)) \quad \leq \quad \sum_{\substack{(T,R)\in \tiles(\dlkb)\,\land \\
                  T  \mbox{ is an $r$-sink }} }N((T,R)) \;.
    \]

    \end{enumerate}
  \end{definition}

  Mosaics as above correctly characterize mixed satisfiability in $\DLLITEBOOL$: 
  \begin{theorem}
  Let $\dlkb$ be a $\DLLITEBOOL$ ontology and let $\Sigma \subseteq
  \conceptnames(\Oo) $. Then the following are equivalent:
  \begin{enumerate}[(A)]
  \item $\dlkb$ has a model $\mcI$ such   $A^{\mcI}$ is finite for each 
    $A\in \Sigma$.
    \medskip
  \item There exists a mosaic $N$ for $(\dlkb,\Sigma)$ as in
    Definition~\ref{def:dllite-mosaic}.
 \end{enumerate}
\end{theorem}

\begin{proof} The proof is analogous to the proof of Theorem~\ref{thm:mosaic}.
\end{proof}

The argument for the following theorem is very similar to the argument
for Theorem~\ref{thm:empt-complexity}, where a \nexptime upper bound
was shown. The key difference here is that the inequation systems that
result from Definition~\ref{def:dllite-mosaic} have only polynomially
many inequations in the size of the input ontology. The only challenge
to obtain an \np upper bound is the fact that individual inequations might
have exponentially many variables. However, using the result
in~\cite{DBLP:journals/bsl/Pratt-Hartmann08} (which is in essence due to~\cite{EISENBRAND2006564}), we know that
we can concentrate on solutions in which only polynomially many
variables take non-zero values. 

\begin{theorem}
  It is  an \np-complete problem to decide, given $(\dlkb,\Sigma)$ with $\dlkb$ in $\DLLITEBOOL$,
  whether $\dlkb$ has a model $\mcI$ such $A^{\mcI}$ is finite for all
  predicates $A\in \Sigma$.
\end{theorem}
\begin{proof}
  To obtain the
  upper-bound, we use a procedure that non-deterministically generates
  a polynomially sized integer inequation system, which then can be
  checked for the existence of a solution in non-deterministic
  polynomial time.  The generation of the inequations has 3 steps,
  which overcome 3 challenges.

  \begin{enumerate}[(A)]
  \item This step deals with the conditional inequalities in point (5)
    of Definition~\ref{def:dllite-mosaic}. The ``conditionals'' can be
    eliminated by guessing a set $G$ of precisely the roles for which
    a sink tile will exists. Given a guess $G$, we take the following inequations:
    \begin{enumerate}[-]
    \item for all $r\in G$, \[\sum_{\substack{\tau \in
            \tiles(\dlkb)\,\land \\
            \tau \mbox{ is an $r$-sink }}} x_{\tau} > 0\]
    \item for all $r\not\in G$,
      \[\sum_{\substack{\tau =(T,R)\in \tiles(\dlkb)\,\land \\
            r\in R}} x_{\tau} =0
      \]
    \end{enumerate}

  \item The above already leads us to an almost ordinary set
    $\mathcal{E}$ of integer inequations. We need to deal with the
    special value $\alnull$. We argue that by using a small guess we
    can obtain an inequation system over pure integers. Note that we
    have a small number of equations: we have at most
    $n=3\times 2\times |\rolenames(\dlkb)|$ inequations (we count that
    the equation from point (7) as two inequations). We guess a set
    $X$ of variables from the equations such that $|X|\leq n$ (we
    in fact need less than that), and for each $x_\tau\in X$ we have
    that the root of $\tau$ does not have a predicate from $\Sigma$.
    Intuitively, we assume that the variables from $X$ will have the
    value $\alnull$. Based the guessed $X$, we can reduce our
    inequation system.  If a variable from $X$ occurs on the
    right-hand-side of an inequation, then we can drop the inequation.
    If a variable from $X$ occurs on the left-hand-side of an
    inequation, but there is no variable in $X$ that occurs on the
    right-hand-side, then we know that the guess of $X$ was wrong
    (there exist no mosaic under such guess of $X$). We have that the
    original system is satisfiable with the special value $\alnull$ iff
    there exists $X$ as above such that the reduced system has a plain
    integer solution.

  \item After guessing the small $G$ and $X$, we can concentrate on an
    ordinary inequation system, but the problem is that it has
    exponentially many variables. However, due to Lemma 3
    in~\cite{DBLP:journals/bsl/Pratt-Hartmann08}, the inequation
    system has a solution iff it has a solution where only a small
    number (polynomial in the number of inequations) of variables have
    non-zero values. That means we can guess those variables, and
    consider an inequation system over those variables only, which is
    small as desired.
  \end{enumerate}

  The lower bound is inherited  from propositional logic. 
\end{proof}

\subsection*{Mixed Satisfiability in Horn DLs with functionality}
By generalizing the cycle reversion technique for finite model reasoning in Horn description logics~\cite{DBLP:conf/esws/Rosati08,DBLP:conf/kr/GarciaLS14}, we will reduce the problem of mixed satisfiability
to (unrestricted) satisfiability. This yields a series of results for Horn DLs with functionality: $\DLLITE_{\mathit{core}}^{\mathcal{F}}$, 
$\DLLITE_{\mathit{Horn}}^{\mathcal{F}}$, and Horn-$\ALCIF$, subsuming $\ELIbot$ with functionality. 
Since Horn-$\ALCIF$ subsumes all the above mentioned DLs, we will present the approach assuming a Horn-$\ALCIF$
ontology $\Oo$.

\smallskip
We start by recalling how the cycle reversion technique works in the case of finite model reasoning in Horn-$\ALCIF$.

\begin{definition}
A \emph{finmod cycle in $\Oo$} is a sequence 
$K_1,r_1, \dots,r_{n-1}, K_n$,
with  $n \geq 1$,  $K_1, \dots, K_n$ conjunctions of concept names, and  $r_0, \dots ,r_n$ roles  such that:
\begin{itemize}
 \item[-] $K_1 = K_n$,
 \item[-] $ K_i \sqsubseteq_\Oo \exists r_i. K_{i+1}$, for  $1 \leq i < n-1$, and 
 \item [-] $K_{i+1} \sqsubseteq_\Oo (\leq 1 \ r^-_{i} \ A)$, for some $A \in K_i$, 
\end{itemize}
where $K \sqsubseteq_\Oo C$ is a short-hand for $\Oo \models K \sqsubseteq C$, for some concept $C$. 

\smallskip \noindent Further, we say that a finmod cycle $K_1,r_1, \dots,r_{n-1}, K_n$ in $\Oo$ is \emph{reversed} if 
$\Oo \models K_{i+1} \sqsubseteq \exists r^-_i. K_i$ and $\Oo \models K_i \sqsubseteq (\leq 1 \ r_{i} \ K_{i+1})$.
 \end{definition}
%
%

We then have the announced  result.

\begin{theorem}[\cite{DBLP:conf/kr/GarciaLS14}]\label{th:finsat}
 Let $\Oo$ be a Horn-$\ALCIF$ ontology. An ontology $\widehat{\Oo} \supseteq \Oo$ can be constructed such that 
 \begin{itemize}
 \item the size of $\widehat \Oo$ is bounded by  $\mn{poly}(|\Oo|)$, 
  \item all finmod cycles  in $\Oo$ are reversed in $\widehat\Oo$,  and  
  \item 
 $\Oo$ is finitely satisfiable iff   $\widehat\Oo$ is satisfiable. 
 \end{itemize}
\end{theorem}
So, in a nutshell, $\widehat \Oo$ above is obtained by adding axioms to $\Oo$ that ensure that all the finmod cycles in $\Oo$ are reversed. 



%

\medskip 
Now, we proceed to extend this result for deciding mixed satisfiability. We start by introducing some additional notions. Let $\Sigma \subseteq \mn{N_C}(\Oo)$, we say that  a model $\Ii$ of $\Oo$ is a \emph{$\Sigma$-finite model of $\Oo$} if $A^\Ii$ is finite for every $A \in \Sigma$,  and that a type $\tau$ is \emph{finitely realized in $\Ii$} if the set of elements in $\Ii$ realizing $\tau$ is finite.   The first step will then be  to compute the set of all types that are finitely realized in some model $\Sigma$-finite model. We will show that to compute this set is enough to compute a set of conjunctions that determine such types. Then, as a second step,  by reversing all cycles in $\Oo$ in which those ``partially specified'' types participate, we will obtain the desired  reduction from mixed satisfiability to satisfiability. 

In what follows, we will deliberately confuse conjunctions of concept names and sets of concept
names in what follows. 

\begin{definition}\label{def:SigmaStar}
Let $\Sigma^*$ be the smallest set of conjunctions occurring in $\Oo$ such that: 
\begin{itemize}
 \item $K \in \Sigma^*$, for every $K \sqsubseteq_\Oo A \in \Oo$ such that $A \in \Sigma$;
 \item $K \in \Sigma^*$, for every $K' \in \Sigma^*$ such that 
 $K \sqsubseteq_\Oo \exists r. K'$ and  
 $A \sqsubseteq_\Oo (\leq 1 \ r^- \ K )$. 
 
\end{itemize}

\end{definition}

The following lemma establishes that $\Sigma^*$ precisely determines all the types that have to be finitely realized in a 
$\Sigma$-finite model of $\Oo$.
\begin{lemma}
Let $\tau$ be a realizable $\Oo$-type $\tau$. $\tau$ is finitely realized in  every model 
 $\Sigma$ finite model $\Ii$ of $\Oo$  iff $K \subseteq \tau$ for some $K \in \Sigma^*$. 
 \end{lemma}
 \begin{proof}
  The  direction $(\Leftarrow)$ follows from the definition of $\Sigma^*$ and the semantics of Horn-$\ALCIF$. 
  For the direction $(\Rightarrow)$, we show the contrapositive. So, we  assume that there is no $K \subseteq \tau$ as in the statement of the lemma. Let $\Ii$ be a $\Sigma$-finite model such that $\tau$ is finitely realized in $\Ii$. Then, we can construct a $\Sigma$-finite model 
  $\Jj$ from $\Ii$  such that $\tau$ is not finitely realized by creating countably many copies $d_1, d_2,\ldots $ of some fixed $d \in \adom{\Ii}$ with $\tp_\Ii(d)= \tau$. The construction of $\Jj$ relies on an special kind of unraveling, inductively defined as follows. The initial instance $\Jj_0$ contains a single 'distinguished' element $d_0$ such that $\tp_{\Jj_0}(d_0) = \tau$ and a copy of every element $f$ in $\Ii$ such that 
  $K \subseteq \tp_\Ii(f)$ for some $K \in \Sigma^*$ with  their unary types set as in $\Ii$. 
  We will denote this set of copies as $\Delta_{F}$. For the inductive step, start by setting $\Jj_{i+1}= \Jj_{i}$
  and then extend $\Jj_{i+1}$ as follows.  Include a  fresh copy $d_{i+1}$ of $d$ in $\Jj_{i+1}$, and for every element $e$ in $\Jj_{i}$ such that its type contains some $K$, 
$K \sqsubseteq \exists r. A \in \Oo$ and $e \notin (\exists r. A)^{\Jj_i}$ proceed as follows. Let $e'$ be an element in $\Ii$ witnessing this requirement (such $e'$ exists since $\Ii$ is a model of $\Oo$). 
 Then, if there is a copy of $e'$ in $\Delta_F$,  add an $r$-edge from $e$ to that copy.
  Otherwise,  add a fresh copy of $e'$ to $\Jj_{i+1}$ and an $r$-edge from $e$ to that copy. 

It can be readily  checked that the interpretation obtained in the limit is a model of $\Oo$. To see that it is 
a $\Sigma$-finite model,  it suffices to observe that (1)~$\Delta_F$ is finite, and (2)~ 
the types of the fresh witnesses (which are copied from the original $\Sigma$-finite model $\Ii$) added in the inductive steps do not contain any concept from $\Sigma$, as otherwise they would be in $\Delta_F$.
 \end{proof}
Before describing the generalized cycle reversion, we make a note on the complexity of computing $\Sigma^*$
in all the DLs considered.  
 \begin{proposition} $\Sigma^*$ can be computed
\begin{itemize}  
 \item in NLogSpace for $\DLLITE_{\mathit{core}}^{\mathcal{F}}$,
 \item in PTime for $\DLLITE_{\mathit{Horn}}^{\mathcal{F}}$,
 \item in ExpTime for Horn-$\ALCIF$.
\end{itemize}
 \end{proposition}
\begin{proof}
 The complexity bounds follow from the definition of $\Sigma^*$ and the known complexity of deciding
 $\sqsubseteq_\Oo$ for each of the DLs. 
 \end{proof}

With the computation of $\Sigma^*$ in place, we can now define what we will call \emph{$\Sigma^*$-cycle reversion}. Firs, we need  a straightforward generalization of finmod cycles. 
\begin{definition}\label{def:sigmaCycle}
Let $\Oo$ be a Horn-$\ALCIF$ ontology,   
and  $\Sigma^*$ a set of conjunctions occurring in $\Oo$. A $\Sigma^*$-cycle in $\Oo$ is a sequence
$$K_1,r_1,\dots,K_{n-1}, r_n, K_n \quad \text{such that}$$ 
\begin{itemize}
 \item[-] $K_1 \in \Sigma^*$, 
 \item [-] $K_1 = K_n$, 
 \item[-] $K_i \sqsubseteq_{\Oo} \exists r_i. K_{i+1}$, for  $1 \leq i < n-1$, and 
 \item [-] $K_{i+1} \sqsubseteq_{\Oo} (\leq 1 \ r^-_{i} \ A)$, for some $A \in K_i$. 
\end{itemize}
\end{definition}

For our purpose then it will be enough to reverse all $\Sigma^*$-cycles in $\Oo$ with $\Sigma^*$ as in Definition \ref{def:SigmaStar}. As for finmod cycles in Horn-\ALCIF, we can achieve this using a consequence-driven calculus~\cite{DBLP:conf/kr/GarciaLS14} (see figure~\ref{fig:rules} for reference), by adding a precondition in rule $\mathbf{R9}$ that $K_1 \in \Sigma^*$. The extended  ontology $\widehat \Oo$ in which all $\Sigma^*$-cycles are reversed can be obtained after an exponential number of rule applications.

The following lemma provides the reduction for Horn-$\ALCIF$. 
\begin{lemma}\label{lemma:reduc}
 Let $\widehat \Oo \supseteq \Oo$ be an ontology in which all $\Sigma^*$-cycles in $\Oo$ are reversed . Then 
 $\widehat \Oo$  is satisfiable iff $\Oo$ has a $\Sigma$-finite model. 
\end{lemma}

\begin{proof}
The  ($\Leftarrow$) direction follows since $\Oo \subseteq \widehat \Oo$ and every every rule is sound w.r.t. the $\Sigma$-finite model semantics. 

For the  ($\Rightarrow$) direction, let $\Ii$ be a model of $\Oo$ that realizes at least one type containing 
some conjunction in $\Sigma^*$ (otherwise the statement holds trivially). 
One can construct a finite interpretation $\Jj$ using 
the special unraveling in the proof of Theorem 3 in~\cite{DBLP:conf/kr/GarciaLS14}, by considering only the set of 
types realized in $\Ii$, that contain some conjunction from $\Sigma^*$. 

In order to extend $\Jj$ to a $\Sigma$-model of $\Oo$ it remains to add missing witnesses for existential restrictions in $\Oo$, this can be done as follows.  For every element $d$ in $\Jj$ whose type contains some $K$, such that  $K \sqsubseteq \exists r. A \in \Oo$ and $d \notin (\exists r. A)^{\Jj_i}$, let $e$ be an element in $\Ii$ witnessing this requirement, which exists since $\Ii$ is a model of $\widehat \Oo$. 
If the type of $e'$ contains some conjunction from $\Sigma^*$, then 
then add an $r$-edge from $e$ to some element $d'$ in $\Jj$ with $\tp_\Jj(d') = \tp_\Ii(e')$. This element exists by construction and moreover, adding this $r$-edge does not violates functionality. 

Otherwise, if the type of $e'$ does not contain any conjunction from $\Sigma^*$,
add a fresh element of $d'$ to $\Jj$ and an $r$-edge from $e$ to $d'$ and set $\tp_\Jj(d') = \tp_\Ii(e')$.
\end{proof}

 \newcommand{\Rule}[1]{\textup{\textbf{\textsf{\small\boldmath R#1}}}}
\newcommand{\qnrleq}[3]{\ensuremath{(\leqslant #1 \; #2 \; #3)}}
\newcommand{\qnrgeq}[3]{\ensuremath{(\geqslant #1 \; #2 \; #3)}}
\begin{figure}[t]
  \centering
  \begin{small}

\begin{align*}
&\Rule1~~
\frac{}
{K \sqcap A \sqsubseteq A} 
&&\Rule2~~
\frac{}{K \sqsubseteq \top} 
\\[3mm]
&\Rule3~~
\frac{K \sqsubseteq A_i \ \  \bigsqcap  A_i \sqsubseteq C }{K \sqsubseteq C} 
&&
\Rule4~~
\frac{K \sqsubseteq \exists r. K' \ \ 
   K' \sqsubseteq \forall r^-. A  }{K \sqsubseteq A}   
\\[3mm]
&\Rule5~~
\frac{K \sqsubseteq \exists r . K' \ \  K  \sqsubseteq \forall r. A  }{K \sqsubseteq \exists r . (K'\sqcap A)}
&&
\Rule6~~
\frac{K \sqsubseteq \exists r . K' \ \ K' \sqsubseteq \bot }{K \sqsubseteq \bot}
\end{align*}

\begin{align*}
&\Rule7~~
\frac{\begin{array}{r@{~}c@{~}l@{\quad}r@{~}c@{~}l@{\qquad}r@{~}c@{~}l}
K 	& \sqsubseteq & \exists r. K_1 & K   & \sqsubseteq & \exists r . K_2 & K & \sqsubseteq & \qnrleq 1 r A    \\[0.5mm]
 K_1 & \sqsubseteq & A    				& K_2 & \sqsubseteq & A &                      & &   
\end{array}}{ \begin{array}{r@{~}c@{~}l} K  &\sqsubseteq& \exists r. (K_1 \sqcap K_2) \end{array}} 
%
\\[3mm]
%
&\Rule8~~
\frac{\begin{array}{r@{~}c@{~}l@{\quad}r@{~}c@{~}l@{\quad}r@{~}c@{~}l}
K  & \sqsubseteq & \exists r. K'     &  K   & \sqsubseteq &A  & K' & \sqsubseteq & \qnrleq 1 {r^-} A  \\[0.5mm]
K' & \sqsubseteq & \exists r^-. K_1  & K_1 & \sqsubseteq &A &      &                               
\end{array}}
 {K \sqsubseteq \exists (r_1 \cup r_2^-). K' \qquad\quad K  \sqsubseteq   A_1 \quad \text{for any~} A_1 \in K_1 }
\\[3mm]
& \Rule9~~ 
\frac{
  \begin{aligned}
   K_i 		&\sqsubseteq  \exists r_i . K_{i +1}
   &K_{i +1}	&\sqsubseteq	 \qnrleq 1 {r^-_{i}} {A_i},       	
   &   K_{i} 	&\sqsubseteq 	A_i  
  \end{aligned}
  }
  {K_1 \sqsubseteq \exists r_0^-. K_0 \quad\quad K_0 \sqsubseteq \qnrleq 1 {r_0} {A_1} 
} \quad \begin{aligned} & i < n \\  & r_i \in r_i \\ & K_0 = K_n \end{aligned} 
\end{align*}
\end{small}
  \caption{Inference Rules}
  \label{fig:rules}
\end{figure}

We then obtain the following result, where the upper bound follows from Lemma~\ref{lemma:reduc} above  and the fact (unrestricted) satisfiability in Horn-$\ALCIF$ is \exptime-complete. The lower bound follows from the complexity of satisfiability in $\mathcal{ELI}_\bot$.

\begin{theorem}
 Mixed satisfiability in Horn-\ALCIF is \exptime-complete
\end{theorem}

\noindent
To obtain the complexity bounds for the other considered DLs, first note that Definition~\ref{def:sigmaCycle} applies to all of them. Then, note that for instance, in $\DLLITE_{\mathit{core}}^{\mathcal{F}}$ conjunction  considered for cycles are  either empty i.e., $\top$ or have a single concept name, and functionality assertions $\func(r)$ are equivalent to having the axiom $ \top  \sqsubseteq (\leq 1 \ r \ \top)$. Finally, note that for $\DLLITE_{\mathit{core}}$ ontologies, which are not able to express conjunction on the left side of axioms $\Sigma^*$ can be regarded simply as a superset of $\Sigma$.

Next, we assume that the unary type of an element in a model $\Oo$ indicates whether it belongs to the domain and range of a given role $r$. This consideration is  without loss of generality as every ontology $\Oo$ can be transformed into such an equisatisfiable ontology, by adding      
for every role  $r$ name (and its inverse $r^-$) occurring in $\Oo$ the following  axioms 
$A_r  \equiv  \exists r. \top $ and $A_{r^-} \equiv \exists r^-.  \top$. 
Now, identifying a $\Sigma^*$-cycle in $\DLLITE_{\mathit{Horn}}^{\mathcal{F}}$ can be done by representing the ontology as a directed graph in which the nodes are conjunctions of concept names occurring in the ontology,  and  where there is an edge between $K, K'$ if $K \sqsubseteq_{\Oo} \exists r$ and $\exists r^- \sqsubseteq_\Oo K'$  with $\func(r) \in \Oo$.

In the case of $\DLLITE_{\mathit{core}}^{\mathcal{F}}$ the graph is simpler. The nodes are the concept names representing the domain and range of roles in  the ontology, with an edge between $A_r$ and $A_{s^-}$ if 
$A_r \subseteq_\Oo \exists s^-$ and $\func(s^-) \in \Oo$.

Now, note that deciding $\sqsubseteq_\Oo$ can be done in polynomial time on the size of the ontology for  $\DLLITE_{\mathit{Horn}}^{\mathcal{F}}$ and in NLogSpace for $\DLLITE_{\mathit{core}}^{\mathcal{F}}$. Further,  deciding whether an edge in a directed graph belongs to a cycle in LogSpace. We then obtain the following result.

\begin{proposition}~\label{prop:litecycles}
 An ontology $\widehat \Oo \supseteq \Oo$ in which all $\Sigma^*$-cycles are reversed can be constructed 
 
 \begin{itemize}
  \item in \ptime,  for a $\DLLITE_{\mathit{Horn}}^{\mathcal{F}}$ ontology $\Oo$, and
  \item in \nlogspace for a $\DLLITE_{\mathit{core}}^{\mathcal{F}}$ ontology $\Oo$. 
 \end{itemize}

\end{proposition}

Proposition~\ref{prop:litecycles} and Lemma~\ref{lemma:reduc} yield the desired result:
\begin{corollary}
 Mixed satisfiability  is \ptime-complete for $\DLLITE_{\mathit{Horn}}^{\mathcal{F}}$, and complete for \nlogspace for $\DLLITE_{\mathit{core}}^{\mathcal{F}}$.  
\end{corollary}


%
%
%
%
%
%
%
